\documentclass[letterpaper,11pt]{article}
\pdfoutput=1
\usepackage{geometry}
\usepackage{setspace}
\usepackage[ascii]{inputenc}
\usepackage[bookmarksnumbered,linktocpage,hypertexnames=false,colorlinks=true,linkcolor=blue,urlcolor=blue,citecolor=magenta,anchorcolor=green,breaklinks=true,pdfusetitle=true]{hyperref}
\usepackage{bbm,microtype,mathrsfs,amsmath,amssymb,xcolor,amsthm,amsfonts,latexsym,mathtools,graphicx,fix-cm}
\usepackage[shortlabels]{enumitem}
\usepackage{bm,xspace,mathdots,ellipsis,float,caption,subcaption,mleftright}
\usepackage[T1]{fontenc}
\usepackage[english]{babel}
\usepackage[capitalize,nameinlink]{cleveref}
\crefname{inequality}{Ineq.}{Ineqs.}
\creflabelformat{inequality}{#2\textup{(#1)}#3}

\DeclareMathAlphabet{\mathsfit}{T1}{\sfdefault}{\mddefault}{\sldefault}
\SetMathAlphabet{\mathsfit}{bold}{T1}{\sfdefault}{\bfdefault}{\sldefault}
\newtheorem{theorem}{Theorem}[section]

\newtheorem{lemma}[theorem]{Lemma}
\newtheorem{proposition}[theorem]{Proposition}
\newtheorem{corollary}[theorem]{Corollary}

\theoremstyle{definition}
\newtheorem{definition}[theorem]{Definition}

\newtheorem{remark}[theorem]{Remark}

\AddToHook{env/lemma/begin}{\crefalias{theorem}{lemma}}
\AddToHook{env/proposition/begin}{\crefalias{theorem}{proposition}}
\AddToHook{env/corollary/begin}{\crefalias{theorem}{corollary}}
\AddToHook{env/definition/begin}{\crefalias{theorem}{definition}}
\AddToHook{env/remark/begin}{\crefalias{theorem}{remark}}
\AddToHook{env/example/begin}{\crefalias{theorem}{example}}
\numberwithin{equation}{section}
\DeclareMathOperator{\tr}{tr}
\DeclarePairedDelimiter\abs{\lvert}{\rvert}
\DeclarePairedDelimiter\norm{\lVert}{\rVert}

\DeclarePairedDelimiter\parens{\lparen}{\rparen}
\DeclarePairedDelimiter\braces{\lbrace}{\rbrace}

\DeclarePairedDelimiter\bra{\langle}{\rvert}
\DeclarePairedDelimiter\ket{\lvert}{\rangle}
\DeclarePairedDelimiterX\braket[2]{\langle}{\rangle}{#1\delimsize\vert\mathopen{}#2}
\DeclarePairedDelimiterX\ketbra[2]{\lvert}{\rvert}{#1\delimsize\rangle\delimsize\langle\mathopen{}#2}
\renewcommand{\Re}{\operatorname{Re}}
\renewcommand{\Im}{\operatorname{Im}}
\newcommand{\expect}{\mathop{\mathbb{E}}}

\newcommand{\R}{\mathbb{R}}
\newcommand{\C}{\mathbb{C}}

\newcommand{\N}{\mathbb{N}}

\newcommand{\Id}{\mathbbm{1}}

\newcommand{\QHE}{\ensuremath{\mathsf{QHE}}\xspace}
\newcommand{\MIP}{\ensuremath{\mathsf{MIP}}\xspace}
\newcommand{\RE}{\ensuremath{\mathsf{RE}}\xspace}

\newcommand{\Enc}{\ensuremath{\mathsf{Enc}}\xspace} 
\newcommand{\Dec}{\ensuremath{\mathsf{Dec}}\xspace}
\newcommand{\Gen}{\ensuremath{\mathsf{Gen}}\xspace}

\newcommand{\Eval}{\mathsf{Eval}}
\newcommand{\sk}{\mathsf{sk}}
\newcommand{\SK}{\mathsf{SK}}
\newcommand{\CT}{\mathsf{CT}}
\newcommand{\M}{\mathsf{M}}
\newcommand{\ct}{\mathsf{ct}}

\newcommand{\mcA}{\mathcal{A}}
\newcommand{\mcB}{\mathcal{B}}
\newcommand{\mcC}{\mathcal{C}}

\newcommand{\mcG}{\mathcal{G}}
\newcommand{\mcGcomp}{\mathcal{G}_{\textnormal{comp}}}
\newcommand{\mcGseq}{\mathcal{G}_{\textnormal{seq}}}
\newcommand{\Scomp}{S_{\textnormal{comp}}}
\newcommand{\mcH}{\mathcal{H}}
\newcommand{\mcI}{\mathcal{I}}

\newcommand{\mcO}{\mathcal{O}}

\newcommand{\mcS}{\mathcal{S}}

\newcommand{\msA}{\mathscr{A}}
\newcommand{\msAPOVM}{\mathscr{A}_{\textnormal{POVM}}}
\newcommand{\msB}{\mathscr{B}}

\newcommand{\ot}{\otimes}

\renewcommand{\varepsilon}{\epsilon}

\allowdisplaybreaks[4]
\usepackage{tikz}
\usetikzlibrary{quantikz2}

\usepackage[noblocks]{authblk}
\usepackage{caption}
\usepackage{subcaption}

\renewcommand{\paragraph}[1]{\subsubsection*{#1}}

\begin{document}
\title{A bound on the quantum value of \texorpdfstring{\\}{} all compiled nonlocal games\footnotetext{A short version of this work announcing results with no proofs is available in the \emph{Proceedings of the 57th Annual ACM Symposium on Theory of Computing (STOC'25)}.}}
\date{}
\author[1]{Alexander Kulpe\thanks{\url{alexander.kulpe@rub.de}}}
\author[2]{Giulio Malavolta\thanks{\url{giulio.malavolta@hotmail.it}}}
\author[3]{Connor Paddock\thanks{\url{cpaulpad@uottawa.ca}}}
\author[1]{Simon Schmidt\thanks{\url{s.schmidt@rub.de}}}
\author[1]{Michael Walter\thanks{\url{michael.walter@rub.de}}}
\affil[1]{Ruhr University Bochum}
\affil[2]{Bocconi University}
\affil[3]{University of Ottawa}
\hypersetup{pdfauthor={Alexander Kulpe, Giulio Malavolta, Connor Paddock, Simon Schmidt, Michael Walter}}
\maketitle

\begin{abstract}
    A cryptographic compiler introduced by Kalai, Lombardi, Vaikuntanathan, and Yang\ (STOC'23) converts any nonlocal game into an interactive protocol with a single computationally bounded prover.
    Although the compiler is known to be sound in the case of classical provers and complete in the quantum case, quantum soundness has so far only been established for special classes of games.
    
    In this work, we establish a quantum soundness result for all compiled two-player nonlocal games.
    In particular, we prove that the quantum commuting operator value of the underlying nonlocal game is an upper bound on the quantum value of the compiled game.
    Our result employs techniques from operator algebras in a computational and cryptographic setting to establish information-theoretic objects in the asymptotic limit of the security parameter.
    It further relies on a sequential characterization of quantum commuting operator correlations, which may be of independent interest.
\end{abstract}
\tableofcontents

\section{Introduction}\label{sec:cmp_intro}
A \emph{nonlocal game} consists of two (or more) cooperative players that interact with a referee.
In the game, the referee samples a question for each player, to which each player replies with an answer.
The referee decides if the players win or lose based on the tuple of questions and answers.
Communication is not permitted between players, hence each player has no information about the questions given to the other players, nor do they know the answers provided to the referee by the other players.
Nevertheless, the description of the game is known to the players ahead of time, allowing them to strategize and maximize their probability of winning the game.
The \emph{classical value}~$\omega_c(\mcG)$ of a nonlocal game $\mcG$ is the maximum winning probability of classical players, while the \emph{quantum value}~$\omega_q(\mcG)$ denotes the maximum winning probability of quantum players sharing a finite amount of quantum resources (such as entangled quantum states, like EPR pairs).
In the quantum setting, the no-communication assumption can either be modeled by (i)~spatially separating the players so that they act on tensor product subsystems or (ii)~requiring that the players' actions commute on the joint system.
While these two conditions are equivalent when the state space is finite-dimensional, they are inequivalent in infinite dimensions.
The \emph{quantum commuting operator value} $\omega_{qc}(\mcG)$ denotes the maximum winning probability over strategies where the players' measurement operators commute.

Nonlocal games have proved highly influential in quantum information.
They have significantly advanced and operationalized our understanding of entanglement~\cite{bell1964einstein}, and offered a productive framework for separating correlations arising from various physical models.
Notably, they have been used to discern the classical and quantum correlations~$C_c\subsetneq C_q$~\cite{chsh}, the quantum correlations and their closure~$C_q\subsetneq C_{qa}$~\cite{Slofstra17}, and the latter from the commuting operator correlations~$C_{qa}\subsetneq C_{qc}$ as a consequence of~$\MIP^*=\RE$~\cite{mipre}.
Along the way, they became an important topic in complexity theory, through their connection to multiprover interactive proofs \cite{CHTW04}, and they allow certifying quantum computation in the two-prover setting \cite{reichardt2013classical,coladangelo2024verifier,grilo:LIPIcs.ICALP.2019.28,mipre}.

\subsection{Background on compiled nonlocal games}

A fundamental question in this area is whether two non-communicating provers are really necessary to build such protocols.
The single-prover setting, where a verifier interacts with a single computationally-bounded prover, is both \emph{theoretically appealing} and \emph{practically motivated} since the no-communication assumption can be difficult to enforce.
A breakthrough by Mahadev showed that quantum computations can be verified in this setting~\cite{mahadev}.
To establish a conceptual connection between the two worlds, Kalai, Lombardi, Vaikuntanathan, and Yang (KLVY) proposed a generic procedure to transform any nonlocal game \cref{sfig:nonlocal_game} into a single-prover protocol, called the \emph{compiled nonlocal game}~\cref{sfig:comp_game}, replacing the no-communication assumption between players with a computational assumption on the prover~\cite{klvy}.
For instance, the KLVY compiler translates a two-player game into a four-round game with a single player (prover) and a referee (verifier).
Questions are asked and answered \emph{sequentially}, rather than in parallel, and the leaking of information to the next round is prevented by cryptographic assumptions.
To achieve the desired functionality the construction employs a quantum homomorphic encryption (QHE) scheme with classical messages \cite{qfhe,brakerskiqfhe}, which is used to encrypt the first question (and consequently the answer) of the prover, whereas the second question is sent in plain.
This results in a \emph{quantum polynomial time (QPT)} assumption on the prover, as any prover with greater computational power could break the security of the QHE scheme, and gain information about the encrypted questions.


In addition to outlining the compilation procedure, \cite{klvy} established \emph{classical soundness} and \emph{quantum completeness} of the compiler.
This means that while classical provers cannot exceed the classical value of the corresponding nonlocal game, quantum provers can even achieve the quantum value, in the asymptotic limit where the security parameter~$\lambda$ (of the underlying quantum homomorphic encryption scheme) tends to infinity.
This implies that any nonlocal game with~$\omega_c < \omega_q$ can be converted into a protocol that classically verifies quantum advantage.

\begin{figure}[t]
\centering
\begin{subfigure}[b]{0.3\textwidth}
\centering
\includegraphics[height=2.75cm]{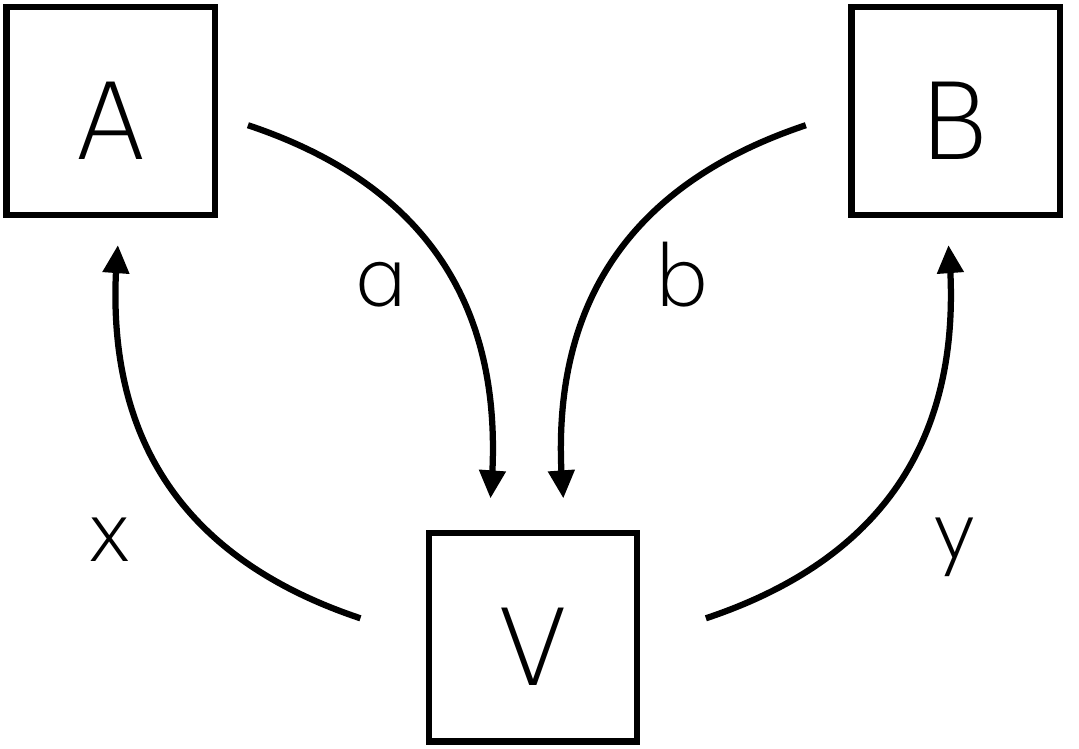}
\caption{Two-player nonlocal game}
\label{sfig:nonlocal_game}
\end{subfigure}
\hfill
\begin{subfigure}[b]{0.3\textwidth}
    \includegraphics[height=2.75cm]{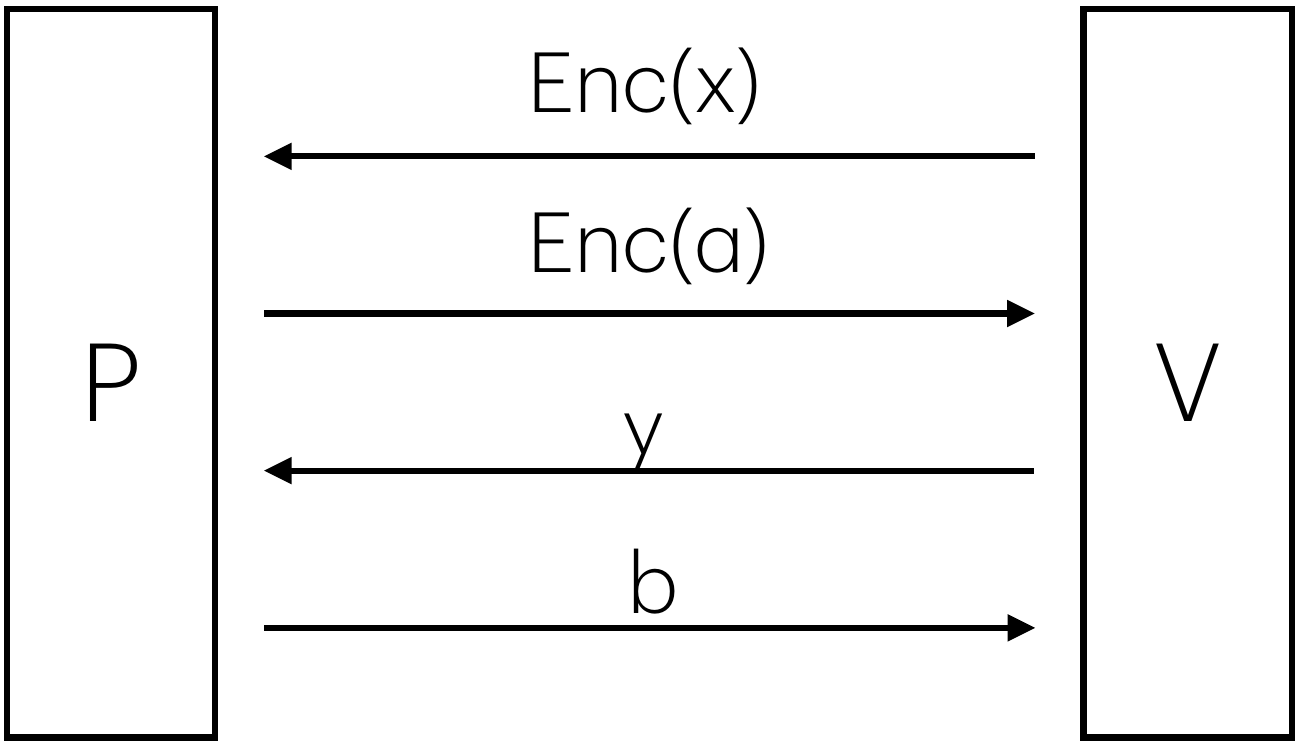}
    \caption{Compiled nonlocal game~\cite{klvy}}
    \label{sfig:comp_game}
\end{subfigure}
\hfill
\begin{subfigure}[b]{0.3\textwidth}
    \includegraphics[height=2.75cm]{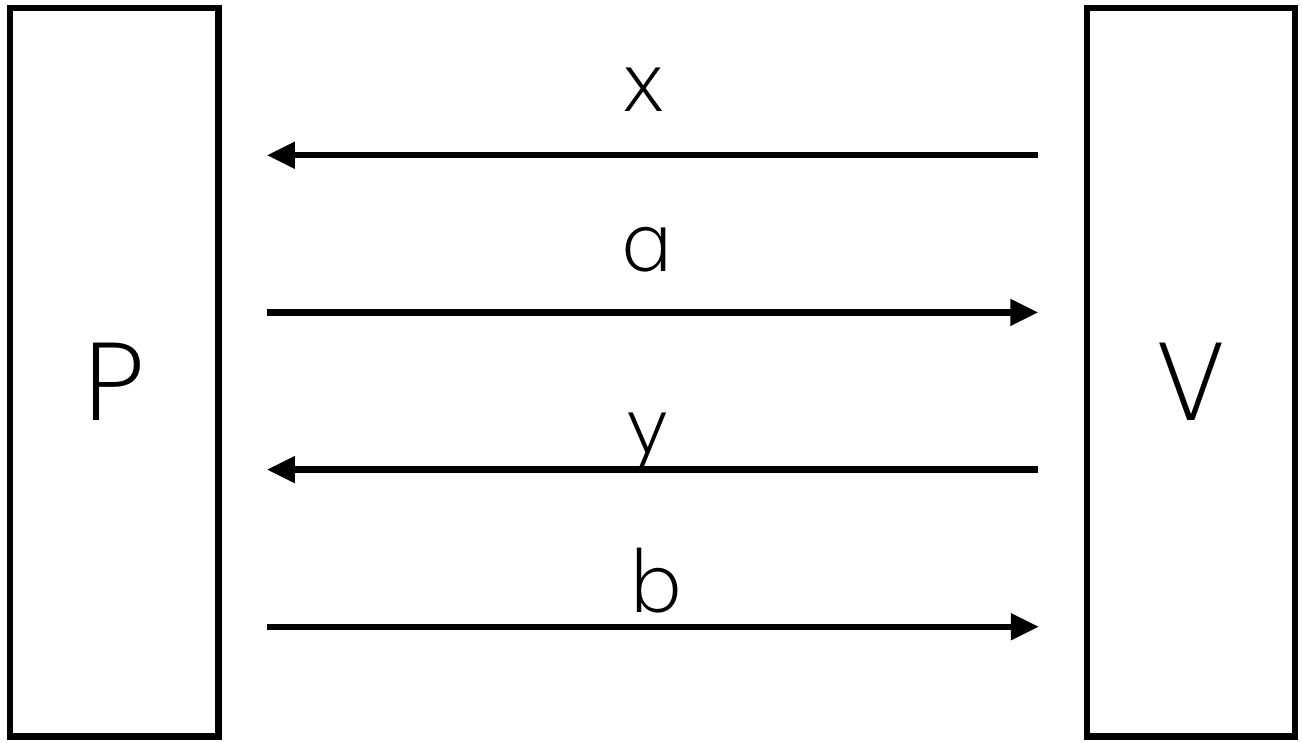}
    \caption{Sequential game without cryptography}
    \label{sfig:seq_game}
\end{subfigure}
\hfill
\caption{\cref{sfig:nonlocal_game} depicts a conventional nonlocal game between a verifier $\mathsf{V}$ and non-communicating computationally unbounded quantum provers $\mathsf{A}$ and $\mathsf{B}$. \cref{sfig:comp_game} depicts a compiled game with verifier $\mathsf{V}$ and a single computationally bounded quantum prover~$\mathsf{P}$. The security of the encryption scheme ensures that post measurement states following the first round of communication are computationally indistinguishable. \cref{sfig:seq_game} depicts a sequential game, which is an idealized compiled game, with a verifier $\mathsf{V}$ and a single quantum prover $\mathsf{P}$. In a sequential game we assume that all post-measurement states after the first round of communication are identical.}
\label{fig}
\end{figure}

However, establishing \emph{quantum soundness}, i.e., an upper bound on the success probability that can be achieved by QPT strategies for the compiled game, has proven to be a more difficult task. Recent works have achieved a bound only for special classes of games, such as the CHSH game~\cite{natarajan2023bounding}, the class of XOR games~\cite{xor,baroni2024quantum}, tilted-CHSH scenarios~\cite{mehta2024self}, and for self-tests on Pauli measurements on maximally entangled states~\cite{metger2024succinct}. As a consequence of these works, compiled nonlocal games paved the way for a conceptually elegant and modular way to perform verification of BQP and QMA computations with a classical verifier.
Despite this progress, establishing a general bound on the quantum value that applies to \emph{all} compiled nonlocal games remained \emph{elusive}.




\subsection{Main results}\label{sub:main_results}
In this work, we make progress on this fundamental open problem by showing that the quantum value of the compiled game is always upper-bounded by the quantum commuting operator value of the nonlocal game, for any arbitrary two-player nonlocal game.
To achieve this we also introduce some new~conceptual~tools.

\begin{theorem}\label{thm:main1}
For large enough security parameter~$\lambda$, no QPT strategy can win the compiled game with probability exceeding the quantum commuting operator value of the game by any constant.
\end{theorem}
\noindent
In other words, we show that for any two-player nonlocal game~$\mathcal{G}$ and for any quantum polynomial-time (QPT) strategy~$S$ for the compiled game $\mcGcomp$,
\begin{align*}
  \limsup_{\lambda\to\infty} \omega_\lambda(\mcGcomp,S) \leq \omega_{qc}(\mcG),
\end{align*}
where $\omega_\lambda(\mcGcomp,S)$ denotes the winning probability of the strategy~$S$ as a function of the security parameter~$\lambda$ (\cref{def:comp_game}).
Thus, our theorem establishes quantum commuting operator soundness for all compiled two-player nonlocal games.
The proof of \cref{thm:main1} follows directly from \cref{thm:compiled-game-bound,cor:comp_val_bound}.


At a technical level, our proof combines methods from operator algebras in the context of cryptographic protocols. Our proof builds on two technical ingredients, that might be of independent interest:
\begin{enumerate}[noitemsep,label={(\roman*)}]
\item a \emph{sequential} characterization of quantum commuting operator correlations, motivated in an idealization of the security guarantee offered by the KLVY compiler~\cref{sfig:seq_game}, and
\item the idea of analyzing cryptographic protocols \emph{with computational security} by taking the limit of the security parameter to infinity to obtain a protocol that offers \emph{information-theoretic security}.
\end{enumerate}
Our characterization generalizes a result in~\cite{NCPV12} 
that describes the spatial quantum correlations in terms of so-called ``quansal'' correlations.
It also solves an open problem in the context of relating Bell and prepare-and-measure scenarios and strategies, characterizing the image of~$C_{qc}$ under the mapping defined in~\cite{wright2023invertible}. 
The characterization of commuting operator correlations was also established independently in the context of steering in~\cite{banacki2023steering}.

It is instructive to compare our general bound with prior works on compiled nonlocal games that established quantum soundness in special cases -- first for the CHSH game and then for general XOR games~\cite{natarajan2023bounding,xor,baroni2024quantum}.
Although these results appear to give tighter bounds on the compiled value, bounding it by the quantum value~$\omega_q$ rather than the commuting operator value~$\omega_{qc}$, this is not in fact the case -- by a result of Tsirelson~\cite{Tsirelson87}, we have~$\omega_q = \omega_{qc}$ for the class of XOR games!
This coincidence is also apparent in the SOS proof techniques in~\cite{natarajan2023bounding,xor,baroni2024quantum}, which naturally bound the commuting operator value rather than the quantum one, but so far have resisted generalization to general games.

Self-testing is a powerful primitive in many applications of nonlocal games, as it allows inferring the prover's strategy from the observed correlations (up to unobservable degrees of freedom).
It has both theoretical and practical significance.
For example, a self-testing result for the compiled CHSH game was key ingredient for BQP verification~\cite{natarajan2023bounding}.

While self-testing is conventionally defined in terms of local dilations~\cite{mayers2004self}, this is not suitable for the commuting operator setting.
Following \cite{paddock2023operator}, we say a nonlocal game~$\mcG$ is a \emph{commuting operator self-test} if any optimal commuting operator strategy gives rise to the same expectation values, for all polynomials in the POVM elements.
The latter can be conveniently captured in terms of a state on the max-tensor product of abstract POVM algebras (\cref{def:com selftest}).
We can then prove the following asymptotic self-testing result in the compiled setting:

\begin{theorem}\label{thm:main2}
Let $\mcG$ be a commuting operator self-test.
If a QPT prover wins the compiled game with probability tending to the commuting operator value, then for any question-answer pair for Alice the expectation values of all polynomials in Bob's POVMs are uniquely determined as~$\lambda\to\infty$.
\end{theorem}

\noindent
In other words, we show that if $\lim_{\lambda\to\infty} \omega_\lambda(\mcGcomp,S) = \omega_{qc}(\mcG)$ then
\begin{align*}
    \lim_{\lambda\to\infty} \tr\mleft( \sigma_{xa}^{\lambda} \, P(\{B^{\lambda}_{yb}\}) \mright)
\end{align*}
is uniquely determined by the optimal commuting operator strategy, for every~$x\in\mcI_A,a\in\mcO_A$ and for every noncommutative polynomial~$P$.
Here, $\sigma_{xa}^\lambda$ denotes the state of the QPT strategy after the interactions corresponding to question~$x$ and answer~$a$, and~$B_{yb}$ are the POVM elements for question~$y$ and answer~$b$ applied subsequently (cf.\ \cref{fig}).
The proof of \cref{thm:main2} follows from the second statement in \cref{thm:com self test}.

While the preceding general discussion is necessarily somewhat abstract, its consequences are concrete.
For example, the CHSH game is a commuting operator self-test and it is known that in any optimal strategy the Bob observables anti-commute on the state.
Thus our theorem shows that in any asymptotical optimal strategy, it holds that $\lim_{\lambda\to\infty} \tr\mleft( \sigma_{xa}^{\lambda} \, \{B^\lambda_0,B^\lambda_1\}^2 \mright) = 0$, where~$B^\lambda_x := B^\lambda_{x0} - B^\lambda_{x1}$.
In this way we recover as a special case a version of the self-testing result established and used to great effect in~\cite{natarajan2023bounding}.

\subsection{Technical outline}


To set some context for our results, let us first recall the basics of the KLVY compiler \cite{klvy}.
Given a two-player nonlocal game $\mcG$ the KLVY transformation turns it into a single-player compiled game $\mcGcomp$ as follows:
\begin{itemize}
    \item The referee sends to the player an encryption of what would be Alice's question $\xi = \mathsf{Enc}(x)$, under a QHE scheme.
    \item The prover homomorphically computes an answer and returns the corresponding ciphertext $\alpha= \mathsf{Enc}(a)$ to the referee.  (In the honest case, this is Alice's encrypted answer)
    \item The referee sends Bob's question $y$ in the plain.
    \item The player returns an answer $b$. (In the honest case, this is Bob's answer)
\end{itemize}
It is clear that the quantum value of the compiled game is at least the quantum value of $\mcG$, since the player can run the optimal strategy for $\mcG$, by the homomorphic properties of QHE. The more challenging aspect is to prove an \emph{upper bound} on the quantum value of $\mcGcomp$, which is the focus of this work.


Let us start with a simple observation about the compiled game.
Since Bob's question is only revealed \emph{after} Alice's question was asked and answered, the marginal distribution of Alice's answer must be independent of Bob's question. In other words, any player's strategy must always be \emph{non-signaling} from Bob to Alice.

For the other direction, we would ideally like to say that the same holds, because of the security of the encryption scheme, which leaks no information about Alice's question to the player. Unfortunately, this is not quite true and we rather have a weaker form of \emph{computational non-signaling}, i.e.,  the state of the player is indistinguishable from a (possibly different) state that is truly independent of Alice's question. To make this more precise, let $\lambda$ be the security parameter, $\sigma_x^\lambda, \sigma_{x^\prime}^\lambda$ be the states after Alice's computation for question $x$ and $x^\prime$, respectively, and let $\{B_{yb}^\lambda\}$ be the POVM element corresponding to the player's behavior in the last two rounds (Bob's computation) for question $y$ and answer $b$.
Then computational
non-signaling says that:
\[
\tr(\sigma_x^\lambda B_{yb}^\lambda) \approx \tr(\sigma_{x^\prime}^\lambda B_{yb}^\lambda)
\]
up to negligible factors. In fact, we can push this observation one step further, by noting that the parallel/sequential combination of efficient algorithms is once again efficient, and thus the above guarantee should also hold for such algorithms. In more details, block encodings of efficient POVMs behave well under linear combination and multiplication allowing us to make a stronger statement about computational non-signaling.
We can show that for all (constant-degree) polynomials $P = P(\{B_{yb}\})$ of non-commuting variables $\{B_{yb}\}$, we have:
\[
\tr(\sigma_{x}^\lambda P(\{B_{yb}^\lambda\})) \approx \tr(\sigma_{x^\prime}^\lambda P(\{B_{yb}^\lambda\})).
\]
We refer to this property as \emph{computationally strong non-signaling}. By generalizing results from \cite{natarajan2023bounding,xor} from observables to polynomials of POVMs, we can prove that the security of the QHE scheme guarantees that the strategies in the compiled game are computationally strong non-signaling. The formal result is stated in \cref{prop:polysecurity}.

To make sense of why this generalized property is useful, let us first consider an idealized scenario without encryption (and hence no security parameter $\lambda$), where we have an exact equality $\tr(\sigma_{x} P(\{B_{yb}\})) = \tr(\sigma_{x^\prime} P(\{B_{yb}\}))$. In other words, we consider an (unrealistic and) idealized situation of a \emph{sequential game} (cf.~\cref{fig}) that is exactly \emph{strongly non-signaling}, i.e., where the player's state after the first two rounds is independent of $x$. Dealing with the additional negligible factor will be an important part of the challenge, but we will come back to it at a later point in this overview.

We claim that any strategy for the sequential game that is guaranteed to be (exact) strongly non-signaling, is in fact a commuting operator strategy in the corresponding nonlocal game. We show a sketch of this proof, for the simplified case of finite-dimensional strategies. This will not be sufficient for our main result (for reasons that will become clear later), but it is useful to gain some intuition on why this statement is actually true.

By definition of strong non-signaling, we have that the (mixed) states $\sigma_x$ after the first two rounds of interaction are independent of $x$, meaning that those states are all the same. Let $\sigma := \sigma_x$ for all $x$ denote this state. Let further $\sigma_{xa}$ be the (subnormalized) state acting on a finite-dimensional Hilbert space $\mcH$ after the first two rounds of interaction, where $x$ is Alice's question and $a$ is Alice's answer. Consider an arbitrary purification $\ket{\psi_{xa}}$ of this state, we consequently obtain that:
\[
\ket{\psi_x} = \sum_{a\in \mcO_A} \ket{a} \ot \ket{\psi_{xa}}
\]
where $\mcO_A$ is the set of possible Alice answers. Since these states are all purifications of the same state $\sigma$, by Uhlmann's theorem, they are all related by a unitary operation acting only on Alice's register. We can use this fact to construct a quantum correlation as follows: Let $\ket{\psi_{x_0}}$ as above, for some fixed $x_0 \in \mcO_A$, be the joint state shared between Alice and Bob. While the Bob operator stays the same as in the sequential game, Alice first applies a unitary $U_x$ to switch to the purification $\ket{\psi_x}$ and then applies the projection $P_a = \ket{a}\bra{a}$.
This gives us a quantum correlation:
\begin{align*}
    p(a,b|x,y) = \tr(\sigma_{xa}B_{yb}) = \bra{\psi_{x_0}} U_xP_aU_x^* \ot B_{yb}\ket{\psi_{x_0}}
\end{align*}
which is identical to the correlation in the sequential game. The formal result is stated in \cref{prop:fin dim seq,prop:quantum strong nonsig spatial}.

As hinted above, this proof is not sufficient for us, precisely because the correlations are not \emph{exactly} non-signaling, but rather \emph{computationally} non-signaling. One may wishfully hope that, taking the limit of the security parameter to infinity, one gets rid of the extra negligible function, allowing to execute the above strategy. Unfortunately, this does not appear to be the case, since Bob POVM's do not necessarily generate the whole algebra of bounded operators on a Hilbert space, and therefore the Hilbert space at its limiting point may not even be well-defined.
Instead, what we rather get is that $\sigma_x$ and $\sigma_{x^\prime}$ are indistinguishable for polynomials of Bob POVM elements, i.e., it holds that $\tr(\sigma_xP(\{B_{yb}\})) = \tr(\sigma_{x^\prime}P(\{B_{yb}\}))$ for all Alice questions $x,x^\prime$. What then \emph{does} hold, is that Bob's POVMs (and in particular non-commutative polynomials of Bob's POVMs) generate a \emph{subalgebra} of the algebra of bounded operators on some Hilbert space. In operator theory language, this is a so-called \emph{$C^*$-algebra}.

Among other things, the framework of $C^*$-algebras offers a model of states and observables in quantum mechanics.
In our context, this means that every strategy of the compiled game can be captured by $C^*$-algebras: The first two rounds by partial states $\phi_{xa}$ on a $C^*$-algebra $\msB$, and the last two rounds by POVM elements $B_{yb}$ in $\msB$.
Such an \emph{algebraic strategy} naturally gives rise to correlations $p(a,b|x,y) = \phi_{xa}(B_{yb})$.
The definition of strong non-signaling also extends to the algebraic setting by requiring that the states $\phi_x := \sum_{a \in \mcO_A} \phi_{xa}$ on the $C^*$-algebra $\msB$ are the same for every Alice question $x$, i.e. $\phi_x = \phi_{x^\prime}$. Note that, for the algebra of bounded operators on a Hilbert space $\msB = \mcB(\mcH)$ we recover the previous strongly non-signaling definition.

Guided by this intuition, we can infer that the right statement to prove for the sequential game is that strongly non-signaling \emph{algebraic} strategies characterize the commuting operator correlations.
We show this by relying on two fundamental concepts from the theory of operator algebras:
\begin{itemize}
\item[(i)] To realize an abstract $C^*$-algebra concretely on a Hilbert space we use the \emph{GNS construction}.
\item[(ii)] To put together the commuting operator correlation on the Hilbert space we get from applying the GNS construction, we use the \emph{Radon-Nikodym theorem for $C^*$-algebras} which acts like the purification in the finite-dimensional proof we saw in the beginning.
\end{itemize}




What is left to be done is to move from the idealized setting to the real setting, where the strong non-signaling property holds only up to a negligible (vanishing) summand in the security parameter $\lambda$. In other words, we can only rely on the fact that,  for all $x, x^\prime$, the states:
\[
\sigma^\lambda_x \approx \sigma^\lambda_{x^\prime}
\]
are computationally indistinguishable.
Note that these are operators on different Hilbert spaces for different~$\lambda$, and the same is true for the POVM elements~$\{B^\lambda_{yb}\}$ that correspond to the last two rounds.
To obtain our main result, we have to analyze what happens when one takes the limit as~$\lambda\to\infty$, which we achieve by incorporating mathematical tools from operator algebras.
%

For any QPT strategy, we have a sequence $\{\omega_\lambda\}_{\lambda}$ of winning probabilities for the compiled game $\mathcal{G}_{comp}$, indexed by the security parameter. Since the interval $[0,1]$ is (sequentially) compact and it holds $\omega_\lambda\in [0,1]$ for all $\lambda$, we obtain that this sequence of winning probabilities has a convergent subsequence. Our goal is to prove that no QPT strategy can win the compiled game with probability exceeding the quantum commuting operator value of the game by any constant. To show this result, we have to bound the limit of any convergent subsequence by the quantum commuting operator value. While we get a convergent subsequence from the compactness of the interval $[0,1]$, this is not enough to obtain our bound: The values $\omega_\lambda$ do not keep track of the states and POVM's that are present in the QPT strategy. Therefore, we have to use a more advanced compactness argument to achieve our goal.

Specifically, we show that the essential part of the prover's strategies can be captured by a sequence of states on a single \emph{universal~$C^*$-algebra}, and use a compactness argument to prove the existence of a limiting state. We want to work with the $C^*$-algebra $\msA_\text{POVM}^{\mcI_B, \mcO_B}$ where $\mcI_B$ and $\mcO_B$ are the corresponding finite question and answer sets of Bob. The algebra is generated by elements $\{B_{yb}\}_{y \in\mcI_B, b \in \mcO_B}$ satisfying the relations of a POVM for every $y$ (and in particular the set of non-commutative polynomials of Bob's POVMs is dense in this algebra).
Its most important feature is the following \emph{universal property}: For any collection of POVMs $\{B_{yb}^\lambda\}$ acting on some Hilbert space $\mcH^\lambda$, there is a unique $*$-homomorphism: \[\theta_\lambda: \msA_\text{POVM}^{\mcI_B, \mcO_B} \to \msB(\mcH^\lambda)\] that maps $B_{yb} \mapsto {B_{yb}^\lambda}$ for all $y \in \mcI_B$ and $b \in \mcO_B$.
Recall that the correlations of a QPT strategy are given by $p_{\lambda}(a,b|x,y)=\tr(\sigma_{xa}^\lambda B_{yb}^{\lambda})$ for some subnormalized states $\sigma_{xa}^\lambda$ and POVM's $B_{yb}^{\lambda}$ on some Hilbert space $\mcH^\lambda$. Now, since $\msA_\text{POVM}^{\mcI_B, \mcO_B}$ has the universal property described above, we can define subnormalized states on the algebra by: \[\phi_{xa}^\lambda: \msA_\text{POVM}^{\mcI_B, \mcO_B}\to \C
\quad\text{s.t.}\quad \phi_{xa}^\lambda(\cdot) = \tr(\sigma_{xa}^\lambda \theta_\lambda(\cdot)).\] Note that all maps $\phi_{xa}^\lambda$ have the same domain now, allowing us to use compactness arguments. Since the norm of each such $\phi_{xa}^\lambda$ is bounded by $1$, we can apply the Banach-Alaolgu theorem -- a central compactness theorem in functional analysis -- to show that the set of those $\phi_{xa}^\lambda$ is sequentially compact in the weak-$*$ topology. Therefore any sequence $\{\phi_{xa}^\lambda\}_{\lambda}$ has a convergent subsequence $\{\phi_{xa}^{\lambda_k}\}_{k}$ that converges pointwise to a functional~$\phi_{xa}$.
It can be verified easily that $\phi_x := \sum_{a \in \mcO_A} \phi_{xa}$ are indeed states of the $C^*$-algebra.
The limiting states can be then shown to \emph{precisely} satisfy strong non-signaling, i.e. $\phi_x = \phi_{x^\prime}$! Since the limiting correlation is strongly non-signaling, we deduce from the previous paragraph that the limiting correlation is a commuting operator correlation.
Therefore, the winning probability of the limiting correlation of any subsequence of $\{\omega_\lambda\}_\lambda$ can at most be the commuting operator value of the nonlocal game.
This concludes the proof of our main result.

The line of reasoning in our proof is reminiscent of proofs of completeness for noncommutative optimization hierarchies such as the NPA hierarchy~\cite{navascues2008convergent}.
It is interesting that it also applies to problems in computationally-secure cryptography, and we expect that it can find more uses in this setting.


We conclude this overview with some additional observations on self-testing.
Since $\omega_{q}(\mcG)=\omega_{qc}(\mcG)$ holds for many nonlocal games, our result enables self-testing theorems in the compiled setting. It turns out that our characterization of quantum commuting correlations (\cref{thm:quantum strong nonsig c*}) already implies such results, if we work with an abstract self-testing definition. In the Hilbert space formulation, a self-test is a nonlocal game for which any optimal quantum strategy is the same as a chosen ideal strategy, up to local isometries. The spirit of this definition can be translated to the abstract $C^*$-algebra world where a nonlocal game is called a commuting operator self-test if any commuting operator strategy that achieves the quantum commuting value determines the same state on the $C^*$-algebra $\mcA_\text{POVM}^{\mcI_A, \mcO_A} \ot_{\max} \mcA_\text{POVM}^{\mcI_B, \mcO_B}$ -- the (maximal) tensor product of the algebras generated by the Alice and Bob POVMs.

We prove an asymptotic self-testing statement for the compiled nonlocal game using the abstract self-test formulation. Let $\mcG$ be a nonlocal game that is a commuting operator self-test with corresponding unique state $\Psi$ on $\mcA_\text{POVM}^{\mcI_A, \mcO_A} \ot_{\max} \mcA_\text{POVM}^{\mcI_B, \mcO_B}$. Consider a strongly non-signaling algebraic strategy with subnormalized states $\phi_{xa}$ and POVMs $\{\tilde B_{yb}\}$ that achieves the quantum commuting value. It holds that for every noncommutative polynomial $P$, we have:
\[
    \phi_{xa}(P(\{\tilde B_{yb}\})) = \Psi(A_{xa} \ot P(\{B_{yb}\})).
\]
This is proven by combining the GNS construction with the fact that the nonlocal game is a commuting operator self-test. For every QPT strategy in the corresponding compiled game $\mcGcomp$ with probability tending to the quantum commuting value, we therefore get that:
\begin{align*}
    \lim_{\lambda\to\infty} \tr\mleft( \sigma_{xa}^{\lambda} \, P(\{B^{\lambda}_{yb}\}) \mright)
 = \Psi(A_{xa} \ot P(\{B_{yb}\})).
\end{align*}
This shows uniqueness of the corresponding expectation values, since the state $\Psi$ is unique.

\subsection{Open problems and outlook}
Our main theorem shows that no QPT strategy can win the compiled game $\mcGcomp$ with probability exceeding $\omega_{qc}(\mcG)$ by any constant.
Since we do not know how to compile general commuting operator strategies, there might be a tighter upper bound.
Is~$\omega_{qc}(\mcG)$ the true answer, or is it instead the quantum value~$\omega_q(\mcG)$?
Since compiled games are fundamentally finite-dimensional, the former might seem implausible.
But recall the result by Ozawa~\cite{ozawa2013tsirelson} that describes~$C_{qc}$ as the limit of \emph{approximately commuting} finite-dimensional quantum strategies as the commutation error approaches zero; cf.~\cite{coudron2015interactive}.
Since QPT strategies for compiled games give rise to finite-dimensional strategies that similarly relax an information-theoretic property (strong non-signaling instead of commutation), this makes it not completely implausible to consider the tantalizing possibility that $\omega_{qc}(\mcG)$ might be the correct answer. On the other hand, while the QHE scheme allows for compiling tensor product strategies, there is no straightforward way to compile approximately commuting strategies for the KLVY approach. If it turns out that there is no way in doing so, the quantum value~$\omega_q(\mcG)$ might be an upper bound for the winning probability of any QPT strategy for the compiled game.  

Our results characterize the limit~$\lambda\to\infty$, but can one establish a quantitative bound?
For example, does it converge faster than any polynomial in~$\lambda$?
This is known for compiled XOR games~\cite{xor} and would be useful for applications.
We speculate that proving this for general games would require substantially different techniques.
A particular obstruction is the conjectured undecidability of the gapped quantum commuting operator value problem~\cite{MNY21}.

Classical soundness was shown for nonlocal games with $k\geq 2$ players \cite{klvy}, quantum soundness was shown for CHSH, respectively XOR games, for $k=2$ players \cite{natarajan2023bounding,xor,baroni2024quantum} and
the soundness results in this paper are proven for all nonlocal games with $k=2$ players. We leave it for future work to generalize our soundness results for $k \geq 2$ players.
Furthermore, the KLVY compiler employs a QHE scheme and thus relies on strong cryptographic assumptions.
The security of the QHE scheme itself only plays a small part of the proof of quantum soundness of the KLVY compiler (\cref{lem:block ind}).
Thus, it seems plausible that our analysis will also work for other compilers with weaker cryptographic assumptions.

\subsection{Organization of the paper}
The remainder of the paper is organized as follows:
\Cref{sec:prelim} covers preliminary material.
\Cref{sec:nonlocal-games} provides an overview of nonlocal games, correlations, and the various values of those games in the spatially separated setting.
\Cref{sec:crypto-compiler} details the KLVY compiler and the description of a compiled nonlocal game, including a discussion of the value of a compiled nonlocal game, and proves a key technical result.
\Cref{sec:alternative-characterization} establishes our equivalent characterization of commuting operator correlations (as well as of classical and quantum correlations) in terms of strategies for a sequential game that satisfy a strong non-signaling property.
\Cref{sec:new-magic} contains the proof of our main result, which establishes the upper bound on the quantum value of a compiled nonlocal game, and the proof of our self-testing result that is obtained using the same techniques.

\section{Preliminaries}\label{sec:prelim}
In this section, we recap some preliminaries from mathematics and computer science as well as fix our notation and conventions.
\subsection{Vectors, operators, quantum mechanics}
Let $\mcH$ be a (possibly infinite-dimensional) Hilbert space.
Elements of $\mcH$ are denoted by $\ket{v}\in \mcH$.
The inner product $\braket\cdot\cdot$ is linear in the second argument and induces the norm~$\norm{v} = \sqrt{\braket v v}$.
We denote by~$\mcB(\mcH)$ the set of \emph{bounded (linear) operators} on~$\mcH$.
We let~$\Id$ denote the identity operator, and denote the \emph{adjoint} of an operator~$A\in \mcB(\mcH)$ by~$A^*$.
The norm on $\mcB(\mcH)$ is the \emph{operator norm}~$\norm A = \sup_{\norm v=1} \norm{A v}$.
For $A,B\in \mcB(\mcH)$ the \emph{commutator} is denoted~$[A,B]=AB-BA$.
The \emph{commutant} of a subset~$\mcS \subseteq \mcB(\mcH)$ is the set $\mcS' = \{ B\in \mcB(\mcH):[B,A]=0, \text{ for all }A\in \mcS\}$.

In quantum mechanics, physical systems are often identified with Hilbert spaces~$\mcH$, and the states of the system are identified with positive semidefinite operators~$\rho$ with unit trace, called \emph{density operators}.
A state is called \emph{pure} if the density operator has rank one, and otherwise it is called \emph{mixed}.
Any unit vector~$\ket v \in \mcH$ determines a pure state by the formula~$\rho = \ketbra v v$, and conversely any pure state can be written in this way, hence the two concepts are often identified.
The \emph{trace distance} is the statistical distance between the distributions associated with two density operators~$\rho$ and~$\sigma$ and is given by the formula $\frac12\norm{\rho-\sigma}_1=\frac12\tr(\abs{\rho-\sigma})$, where $\norm\cdot_1$ is the Schatten-$1$ norm and the absolute value of an operator is defined by $\abs A := \sqrt{A^*A}$.
A \emph{measurement} with a finite outcome set~$\mcO$ is described by a collection of bounded operators $\{A_a\}_{a\in \mcO}$ acting on~$\mcH$ such that~$\sum_{a\in \mcO} A_a^*A_a=\Id$.
If the system is in state~$\rho$, then the probability of obtaining outcome~$a$ is given by $p(a)=\tr(A_a^* A_a \rho)$, after which the state of the system is described by~$A_a \rho A_a^* / p(a)$.
The probabilities of measurement outcomes only depend on the operators~$M_a := A_a^* A_a$.
A collection of operators $\{M_a\}_{a\in\mcO}$ such as these which satisfy~$\sum_{a\in \mcO} M_a = \Id$ is called a \emph{POVM}, which is short for positive operator-valued measure, with outcomes in~$\mcO$.
Any POVM arises from a measurement.
\emph{Observables} are self-adjoint elements $B=B^*\in \mcB(H)$, and their \emph{quantum expectation value} with respect to the state~$\rho$ is given by $\tr(\rho B)$.
This can be related to the preceding if one takes~$\mcO$ to be the set of eigenvalues of~$B$ (assuming it is finite) and~$A_a$ as the corresponding spectral projections.
We will often discuss apparatuses with multiple measurement settings, labeled by some index set~$\mcI$, but the same set of outcomes~$\mcO$ for each setting.
This will be denoted by~$\{\{M_{xa}\}_{a\in \mcO} : x\in \mcI\}$, where~$\{M_{xa}\}_{a\in \mcO}$ is a POVM (or measurement) with outcomes in~$\mcO$ for each~$x\in\mcI$.
We often abbreviate and write this as~$\{M_{xa}\}_{a\in \mcO,x\in \mcI}$ when clear from context.

\subsection{Algebras and representations}\label{sec:operator algebra prelims}
An algebra~$\msA$ over the complex numbers is called a \emph{$*$-algebra} if it is equipped with an antilinear involution, which for an element~$A\in\msA$ will always be denoted by~$A^*$, such that~$(AB)^* = B^*A^*$ for all~$A, B \in \msA$.
In this work, every algebra we consider is unital, meaning it contains an identity element~$\Id$.
A \emph{$C^*$-algebra}~$\msA$ is a $*$-algebra that is complete with respect to a submultiplicative norm~$\norm{\cdot}$ that satisfies the $C^*$-identity $\norm{A^*A} = \norm A^2$ for all~$A\in\msA$.
Examples to keep in mind are $\mcB(\mcH)$ and any $*$-subalgebra of it that is closed with respect to the operator norm, with the adjoint and operator norm as defined above.
A more abstract example will be introduced in \cref{sec:new-magic} and serve as a key ingredient to the proof of our main result.
The commutant~$\mcS'$ of any subset~$\mcS = \mcS^* \subseteq \mcB(\mcH)$ is always a $C^*$-algebra (it is even a von Neumann algebra).
An element~$A \in \msA$ is called \emph{positive}, denoted $A \geq 0$, if it can be written in the form~$A = B^*B$ for some~$B\in\msA$.
It is called a \emph{contraction} if $\norm A \leq 1$; when $A$ is positive this can also be stated as~$A \leq \Id$.
A \emph{positive linear functional} on a $C^*$-algebra~$\msA$ is a linear functional $\phi\colon\msA\to\C$ such that~$\phi(A)\geq 0$ whenever~$A\geq 0$.
Positive linear functionals are always bounded: it holds that $\norm\phi=\phi(\Id)$.
Given positive linear functionals $\phi,\psi$, we write $\phi \leq \psi$ to denote that $\phi(A) \leq \psi(A)$ for all $A \geq0$.
A \emph{state} on a $C^*$-algebra $\msA$ is a positive linear functional that is also \emph{unital}, meaning that $\phi(\Id)=1$.

The formalism of $C^*$-algebras generalizes the usual formalism of quantum mechanics outlined above.
For example, any density operator~$\rho$ acting on a Hilbert space~$\mcH$ gives rise to a state $\phi(\cdot) = \tr(\cdot \rho)$ on the $C^*$-algebra $\msA = \mcB(\mcH)$.
The other concepts of quantum mechanics generalize verbatim.
For example, a measurement on~$\msA$ consists of elements~$\{A_a\}_{a \in \mcO} \subseteq \msA$ such that $\sum_a A_a^* A_a = \Id$, and so forth.
The \emph{Gelfand-Naimark-Segal (GNS) construction} shows that, conversely, the abstract world of $C^*$-algebras can always be realized concretely on a Hilbert space.
It asserts that for every $C^*$-algebra~$\msA$ and state~$\phi\colon\msA\to \C$, there exist
a Hilbert space~$\mcH_\phi$,
a $*$-homomorphism $\pi_\phi \colon \msA\to \mcB(\mcH_\phi)$, and
a unit vector $\ket{\nu_\phi} \in \mcH_\phi$
such that
\[ \phi(A) = \bra{\nu_\phi} \pi_\phi(A) \ket{\nu_\phi} \]
for all $A\in \msA$.
Moreover, $\ket{\nu_\phi}$ is cyclic (meaning $\overline{\pi_\phi(\msA) \ket{\nu_\phi}} = \mcH_\phi$) and thereby uniquely determined.
We call $(\mcH_\phi,\pi_\phi,\ket{\nu_\phi})$ a \emph{GNS triple} associated with~$\phi$.
For more information on $C^*$-algebras, we refer the reader to \cite{blackadar2006operator}.


Finally, we recall a result that applies to any normed vector space, but which we will only use for $C^*$-algebras~$\msA$.
The \emph{Banach--Alaoglu theorem} asserts that the unit ball in the dual space, $\{ \phi \colon \msA \to \C : \norm\phi \leq 1 \}$, is compact in the weak-$^*$ topology.
When~$\msA$ is separable, this unit ball is even sequentially compact in this topology, which concretely means the following:
if $\{\phi_n\}_{n\in\N}$ is a sequence of functionals such that $\norm{\phi_n} \leq 1$ for all~$n\in\N$, then there exists a subsequence $\{\phi_{n_k}\}_{k\in\N}$ and a functional~$\phi$ such that $\lim_{k\to\infty} \phi_{n_k}(A) = \phi(A)$ for all $A\in\msA$.


\subsection{Classical and quantum computing}
A function $f \colon \N \to \R$ is called \emph{negligible} if for every $k \in \N$ there exists a $n_0 \in \N$ such that for every $n \geq n_0$ it holds that $f(n) \leq n^{-k}$.
The sum of two negligible functions is negligible.
Unless stated otherwise, numbers are encoded as bitstrings using their binary representation.
To encode a number in unary representation, we use the notation~$1^n$ which refers to the bitstring of length~$n$ that only consists of ones.
We use the notation $x \leftarrow \mu$ to denote that~$x$ is drawn from a probability distribution~$\mu$, and~$x \leftarrow \mcA(y)$ to indicate that~$x$ is obtained by running an algorithm~$\mcA$ with input~$y$.

A \emph{probabilistic polynomial-time (PPT) algorithm} can be described by a probabilistic Turing machine with a polynomial time bound, meaning that there exists a polynomial~$p$ such that for every input $x\in\{0,1\}^*$ the machine halts after at most~$p(\abs x)$ steps.

For quantum computations, we will use the quantum circuit model.
Here, computations correspond to the application of \emph{quantum circuits}, which are unitary operators that operate on the Hilbert space $\mcH = (\C^2)^{\ot k}$ of some number~$k$ of qubits and are given by the composition of unitary gates that each act nontrivially only on (for definiteness) one or two qubits (taken from some fixed universal gate set).
The size of a quantum circuit is the number of gates used in the computation (we assume all qubits are acted upon by at least one gate).
The qubits are typically split into input qubits and auxiliary qubits, which are assumed to be initialized in the $\ket0$ state, unless stated otherwise.
If a classical outcome is desired, a subset of the qubits is measured after the unitary circuit has been applied.
A \emph{quantum polynomial-time (QPT) algorithm} consists of a family of quantum circuits~$\{C_\lambda\}_{\lambda\in\N}$ and a deterministic polynomial-time Turing machine that on input~$1^\lambda$ outputs a description of~$C_\lambda$.
We can often interpret~$\lambda$ as a problem size or as a security parameter.

Any PPT algorithm can be converted into a QPT algorithm (with~$C_\lambda$ a quantum circuit with~$\lambda$ input qubits that when given as input~$\ket x$ and if a suitable number of qubits is measured, implements the same behavior as the PPT algorithm on any bitstring~$x$ of length~$\abs x = \lambda$).

\section{Nonlocal games and strategies}\label{sec:nonlocal-games}
In this section, we briefly review nonlocal games along with the definitions of classical, quantum, and (quantum) commuting operator strategies, correlations, and values for these games.
We also review the definition of non-signaling correlations.
Readers familiar with these concepts may proceed directly to \cref{sec:crypto-compiler}.

\subsection{Nonlocal games}
In the following, let $\mcI_A,\mcI_B,\mcO_A$, and $\mcO_B$ be finite sets, $\mu\colon\mcI_A\times \mcI_B\to \R_{\geq 0}$ be a probability distribution, and $V\colon \mcO_A\times \mcO_B \times \mcI_A\times \mcI_B \to \{0,1\}$ be a function.

\begin{definition}\label{def:nonlocal_game}
A \emph{(two-player) nonlocal game} is a tuple~$\mcG=(\mcI_A,\mcI_B,\mcO_A,\mcO_B,\mu,V)$ describing a scenario consisting of non-communicating players, Alice and Bob, interacting with a referee.
In the game, the referee samples a pair of questions $(x,y)\in \mcI_A\times\mcI_B$ according to~$\mu$, sending question~$x$ to Alice and~$y$ to Bob.
Then, Alice (resp.\ Bob) returns answer~$a$ (resp.\ $b$) to the referee, who computes the rule function~$V$ on the question-answer 
pairs $(a,b,x,y)$ to determine if $V(a,b|x,y)=1$ they win, or $V(a,b|x,y)=0$ they lose.%
\footnote{We use the notation~$V(a,b|x,y)$ instead of~$V(a,b,x,y)$ to emphasize that this represents the value of answers~$a, b$ given questions~$x, y$.}
\end{definition}

All the information about the game $\mcG$ is available to the players before the game.
This allows them to decide on a strategy beforehand.
However, once the game begins the players are not allowed to communicate.
To the referee, the behavior of the players can be modeled by the probabilities~$p(a,b|x,y)$ of answers~$a,b$ given questions~$x,y$ as determined by the strategy.
The collection of numbers~$\{p(a,b|x,y)\}_{a\in\mcO_A,b\in\mcO_B,x\in \mcI_A,y\in \mcI_B} \in \R^{\mcO_A \times \mcO_B \times \mcI_A \times \mcI_B}$ is called a \emph{(bipartite) correlation}.
Thus, the probability of winning the game $\mcG$ under a strategy~$S$, with correlations~$p$, is given by
\begin{equation}\label{eq:omega G S and p}
    \omega(\mcG,S)
=   \omega(\mcG,p)
=   \sum_{x \in \mcI_A, y \in \mcI_B}\sum_{a \in \mcO_A, b \in\mcO_B}\mu(x,y)V(a,b|x,y)p(a,b|x,y).
\end{equation}
Observe that the winning probability of a strategy is simply a linear function of the corresponding correlation that it realizes.

\begin{remark}
Nonlocal games can also be viewed in the context of multiprover interactive proofs.
Here one thinks of the players as provers and the referee as a verifier in an interactive protocol for a language.
The winning probability of the game is the acceptance probability of the verifier.
\end{remark}

\subsection{Strategies and correlations}
One of the main purposes of nonlocal games was to explore the effect of entangled non-communicating players in contrast to classical players (players with no entanglement).
We start with the definition of the latter.

\begin{definition}\label{def:c strat}
A \emph{classical strategy} for a nonlocal game~$\mcG$ consists of
\begin{enumerate}[(i)]
    \item a probability distribution~$\gamma\colon\Omega\to \R_{\geq 0}$ on a (without loss of generality) finite probability space~$\Omega$, along with
    \item probability distributions $\{p_\omega(a|x) : x\in \mcI_A,\omega\in \Omega\}$ with outcomes in~$\mcO_A$ and $\{q_\omega(b|y):y\in \mcI_B,\omega\in \Omega\}$ with outcomes in~$\mcO_B$.
\end{enumerate}
A correlation $\{p(a,b|x,y)\}$ for which there is a classical strategy such that
\begin{equation*}
    p(a,b|x,y)=\sum_{\omega\in \Omega}\gamma(\omega) \, p_\omega(a|x) \, q_\omega(b|y),
\end{equation*}
for all $a\in \mcO_A,b\in \mcO_B,x\in \mcI_A,y\in \mcI_B$ is called a \emph{classical correlation}.
The set of classical correlations is denoted by $C_c(\mcI_A,\mcI_B,\mcO_A,\mcO_B)$ or simply as $C_c$ when the sets $\mcI_A,\mcI_B,\mcO_A,\mcO_B$ are clear from context. It is easy to see that $C_c \subseteq \R^{\mcO_A \times \mcO_B \times \mcI_A \times \mcI_B}$ is a closed convex subset.
\end{definition}

In quantum mechanics, spatially separated subsystems are often represented by the tensor product of Hilbert spaces~$\mcH_A$ and~$\mcH_B$.
The pure states of the joint system are the unit vectors~$\ket\psi\in \mcH_A\otimes \mcH_B$.
Furthermore, if $\{X_a\}_{a\in \mcO_A}$ and~$\{Y_b\}_{b\in \mcO_B}$ are POVMs on~$\mcH_A$ and~$\mcH_B$ respectively, then $\{X_a\otimes Y_b\}_{(a,b)\in \mcO_A\times \mcO_B}$ describes the joint measurement, with outcomes in~$\mcO_A\times \mcO_B$.
With this in mind, we can imagine the players in a nonlocal game to be quantum players described in this way.
The players start out sharing a joint quantum state and, as they are spatially separated and non-communicating once the game begins, any process by which they use the quantum resource in the game can be modelled by POVMs (which can depend on their given question) that the players employ to obtain their answers.
If we assume that the players have finite-dimensional Hilbert spaces at their avail, we arrive at the following definition of a quantum strategy for a nonlocal game.

\begin{definition}\label{def:q strat}
A \emph{quantum strategy} for a nonlocal game $\mcG$ consists of
\begin{enumerate}[(i)]
    \item finite-dimensional Hilbert spaces $\mcH_A$ and $\mcH_B$,
    \item a (without loss of generality) pure quantum state $\ket{\psi}\in \mcH_A\otimes \mcH_B$, along~with
    \item POVMs $\{\{M_{xa}:a\in \mcO_A\}: x\in \mcI_A\}$ acting on $\mcH_A$ and POVMs $\{\{N_{yb}:b\in \mcO_B\}: y\in \mcI_B\}$ acting on $\mcH_B$.
\end{enumerate}
A correlation $\{p(a,b|x,y)\}$ for which there exists a quantum strategy such that
\begin{equation*}
    p(a,b|x,y)=\bra\psi M_{xa} \otimes N_{yb} \ket\psi
\end{equation*}
for all $a\in \mcO_A,b\in \mcO_B,x\in \mcI_A,y\in \mcI_B$ is called a \emph{quantum correlation}.
The set of quantum correlations is denoted by~$C_q(\mcI_A,\mcI_B,\mcO_A,\mcO_B)$ or simply~$C_q$.
\end{definition}

Quantum strategies of particular interest are the \emph{entangled strategies}.
This is because if the state in the quantum strategy is \emph{unentangled}, then the resulting correlation is always classical.
It is easy to see that $C_c \subseteq C_q$, and the inclusion is in general strict, as follows from the existence of nontrivial Bell inequalities~\cite{chsh,Mermin,Peres}.

The restriction to finite-dimensional quantum systems is not the only model.
If one allows the Hilbert spaces~$\mcH_A$ and~$\mcH_B$ to be infinite-dimensional, one gets the set of \emph{spatial} quantum correlations~$C_{qs}$.
Clearly, $C_q \subseteq C_{qs}$, and both sets are convex subsets~$\R^{\mcO_A \times \mcO_B \times \mcI_A \times \mcI_B}$.
It turns out that in general, the inclusion is strict \cite{coladangelo2020inherently} and neither set is closed~\cite{Slofstra17,dykema2019non,coladangelo2020two,musat2020non,beigi2021separation}.
Moreover, both sets have the same closure, denoted by~$C_{qa}$ and named the set of \emph{quantum approximable correlations}.

Another assumption that is not always warranted is the tensor product structure $\mcH = \mcH_A \ot \mcH_B$ of the joint Hilbert space.
For instance, spatially separated quantum systems in quantum field theory need not correspond to a tensor product factorization, but are rather modelled mathematically by commuting subalgebras $\mathscr A, \mathscr B \subseteq \mathcal B(\mcH)$ of observables on a single joint Hilbert space~$\mcH$.
This perspective gives rise to the following class of strategies.

\begin{definition}\label{def:qc_strat}
A \emph{(quantum) commuting operator strategy}%
\footnote{These are also called quantum commuting strategies in parts of the literature.}
 for a nonlocal game $\mcG$ consists of
\begin{enumerate}[(i)]
    \item a Hilbert space $\mcH$,
    \item a (without loss of generality) pure quantum state $\ket{\psi}\in \mcH$, along with
    \item POVMs $\{\{M_{xa}:a\in \mcO_A\}: x\in \mcI_A\}$ and $\{\{N_{yb}:b\in \mcO_B\}: y\in \mcI_B\}$ acting on $\mcH$, such that $[M_{xa},N_{yb}]=0$ for all $a\in \mcO_A$, $b\in \mcO_B$, $x\in \mcI_A$, $y\in \mcI_B$.
\end{enumerate}
A correlation $\{p(a,b|x,y)\}$ for which there exists a commuting operator strategy such that
\begin{equation*}
    p(a,b|x,y) = \bra\psi M_{xa} N_{yb} \ket\psi
\end{equation*}
for all $a\in \mcO_A,b\in \mcO_B,x\in \mcI_A,y\in \mcI_B$ is called a \emph{commuting operator correlation}.
The set of \emph{commuting operator correlations} is denoted by~$C_{qc}(\mcI_A,\mcI_B,\mcO_A,\mcO_B)$~or~$C_{qc}$.
\end{definition}
The set of commuting operator correlations~$C_{qc} \subseteq \R^{\mcO_A \times \mcO_B \times \mcI_A \times \mcI_B}$ is always a closed convex subset~\cite{navascues2008convergent}.
Since every quantum strategy is a commuting operator strategy by properties of the tensor product, it follows that $C_{qa} \subseteq C_{qc}$.
When the Hilbert space $\mcH$ is finite-dimensional then any commuting operator strategy can also be seen as an ordinary quantum strategy, but in general, this is not so.
In fact, there exist commuting operator correlations which have no realization as a quantum strategy, and hence~$C_q\subsetneq C_{qc}$~\cite{Slofstra16}.
Whether the correlation sets~$C_{qa}$ and~$C_{qc}$ were the same became known as Tsirelson's Problem.
This was recently resolved in the celebrated work $\MIP^*=\RE$~\cite{mipre} by the construction of a correlation in~$C_{qc}$ with no realization in~$C_{qa}$.
In turn, this implied a negative resolution to Connes' Embedding Problem, following~\cite{Ozawa2013,fritz2012tsirelson,junge2011connes}.

Given that there are different models of physical correlations, it is interesting to ask for conditions that any correlation should satisfy so that it can reasonably be interpreted as a strategy of non-communicating players.
One such condition is known as the \emph{non-signaling} property: it asserts that the marginal distribution of either player's answers must be independent of the other player's question.
Non-signaling is easily verified to hold for all correlations defined so far.

\begin{definition}
A correlation $p(a,b|x,y)$ is \emph{non-signaling} if for all $x,x' \in \mcI_A$, $y \in \mcI_B$, and $b \in \mcO_B$, it holds that
\begin{align*}
    \sum_{a\in \mcO_A}p(a,b|x,y)&=\sum_{a\in \mcO_A}p(a,b|x',y)
\end{align*}
and moreover for all $x \in \mcI_A$, $y,y'\in \mcI_B$, and $a \in \mcO_A$, it holds that
\begin{align*}
    \sum_{b\in \mcO_B}p(a,b|x,y)&=\sum_{b\in \mcO_B}p(a,b|x,y').
\end{align*}
The set of \emph{non-signaling correlations} is denoted by~$C_{ns}(\mcI_A,\mcI_B,\mcO_A,\mcO_B)$~or~$C_{ns}$.
\end{definition}

To summarize, we have the following inclusion of convex sets, all of which are known to be strict in general:
\begin{align*}
    C_c \subsetneq C_q \subsetneq C_{qs} \subsetneq C_{qa} \subsetneq C_{qc} \subsetneq C_{ns}
\end{align*}
Moreover, $C_c, C_{qa}, C_{qc}, C_{ns}$ are all closed, while $C_q$ and $C_{qs}$ are in general not.

\subsection{Values of games}
One of the original motivations for studying nonlocal games was to understand when different types of strategies achieve different maximal winning probabilities, as this can provide separations between the various correlation sets.
To this end, one defines the value of a game as the highest winning probability for a given class of strategies or, equivalently, correlations.

\begin{definition}
Given a nonlocal game $\mcG$ and $\star \in \{c,q,{qc},{ns}\}$, we define
\begin{align*}
    \omega_\star(\mcG) = \sup_{p \in C_\star} \omega(\mcG, p).
\end{align*}
The quantity $\omega_c(\mcG)$ is called the \emph{classical value} of the game, $\omega_q(\mcG)$ its \emph{quantum value}, $\omega_{qc}(\mcG)$ its \emph{commuting operator value}, and~$\omega_{ns}(\mcG)$ its non-signaling value.
\end{definition}

Clearly, the value only depends on the closure of the corresponding set of correlations.
In particular, $\omega_q$ can equivalently be defined in terms of~$C_{qs}$ or~$C_{qa}$.
Given a nonlocal game~$\mcG$ it is immediate that
\[ \omega_c(\mcG) \leq \omega_q(\mcG) \leq \omega_{qc}(\mcG) \leq \omega_{ns}(\mcG). \]
However, there also exist games~$\mcG$ for which each of these inequalities is strict.
Indeed, the latter is equivalent to the statement that the closed convex sets $C_c \subsetneq C_{qa} \subsetneq C_{qc} \subsetneq C_{ns}$ are in general distinct, as discussed above.%
\footnote{Veritably, the existence of certain nonlocal games $\mcG$ is how several of the strict inclusions $C_c \subsetneq C_{qa} \subsetneq C_{qc} \subsetneq C_{ns}$ were established.}

\section{Compiled nonlocal games}\label{sec:crypto-compiler}
In this section we will review the construction from~\cite{klvy} that allows compiling any multi-player nonlocal game into a single-prover interactive protocol, along with the required cryptography.
We will then prove a technical result that will be key to our later analysis.
It states that as the security parameter tends to infinity, the average state of the prover after the first round of protocol becomes independent of the first round's challenge (question) in a precise sense (\cref{prop:polysecurity}).

\subsection{Quantum homomorphic encryption}
We now define the notion of a quantum homomorphic encryption scheme, which is the central component of the KLVY compiler.
For the purposes of their construction, only classical messages need to be encrypted, and all ciphertexts should be classical.
Moreover, one requires the capability to apply quantum circuits to homomorphically compute on such ciphertexts, and in addition to the classical input and output, these quantum circuits may also act on auxiliary input qubits (which are not encrypted).

Because in this work we do not discuss families of games and their interplay with the security parameter, we can assume that the set of allowed classical messages (which will later correspond to Alice's questions and answers) is a fixed finite set, independent of the security parameter.
We will denote this message set by~$\M$ and assume without loss of generality that it consists of bitstrings of some fixed length~$\ell\in\N$.
Similarly, we may assume that the set of allowed quantum circuits, which we denote by~$\mcC$, is a fixed (but possibly infinite) set independent of the security parameter.
Each circuit~$C \in \mcC$ takes as input some number~$\ell+a_C$ of input qubits, with the first~$\ell$ qubits corresponding to the classical message (encoded in the computational basis) and the remaining~$a_C$ qubits serving as the auxiliary input mentioned above.
In the following definition, we also denote by $\SK$ the set of classical secret keys and by $\CT$ the set of classical ciphertexts; both sets consist of bitstrings.

\begin{definition}\label{def:qhe}
Given sets of classical messages~$\M$ and of quantum circuits~$\mcC$ as above, a \emph{quantum homomorphic encryption scheme} is a tuple
\[ \QHE = \parens{ \Gen, \Enc, \braces{\Eval^C}_{C \in \mcC},\Dec } \]
consisting of algorithms with the following description:
\begin{itemize}
\item
\emph{Key generation:} $\Gen \colon \{1^\lambda\}_{\lambda\in\N} \to \SK$ is a QPT algorithm that takes as input the security parameter~$\lambda$ in unary, and returns a secret key.

\item
\emph{Encryption:}
$\Enc\colon \SK \times \M \to \CT$ is a QPT algorithm that takes as input a secret key and a message, and returns a ciphertext.

\item
\emph{Homomorphic evaluation:}
For every quantum circuit $C\in\mcC$, there is a QPT algorithm $\Eval^C \colon \CT \times (\C^2)^{\ot a_C} \to \CT$ that takes as input a ciphertext and a quantum register on $a_C$ qubits, and returns a ciphertext.

\item
\emph{Decryption:}
$\Dec: \SK \times \CT \to \M$ is a QPT algorithm that takes as input a secret key and a ciphertext, and returns a message.
\end{itemize}
We require that the following two properties hold:
\begin{itemize}
\item \emph{Correctness with auxiliary input:}
Recall that each circuit $C\in\mcC$ acts on a Hilbert space of the form $\mcH_M \ot \mcH_A$, where $\mcH_M = (\C^2)^{\ot\ell}$ and $\mcH_A = (\C^2)^{\ot a_C}$.
For every quantum circuit $C\in\mcC$, for every message~$m\in\M$, for every Hilbert space~$\mcH_B$, and for every quantum state~$\ket\psi_{AB} \in \mcH_A \ot \mcH_B$, there should be a negligible function~$\eta$ of the security parameter such that the states returned by the following two games have trace distance at most~$\eta(\lambda)$, for~all~$\lambda$:
\begin{itemize}
\item Game~1:
Apply $C_{MA} \ot \Id_B$ to $\ket m_M \ot \ket\psi_{AB}$.
Measure register~$M$ to obtain a bitstring~$m'$.
Return $m'$ and register~$B$.
\item Game~2:
Sample a key $\sk \leftarrow \Gen(1^\lambda)$ and encrypt using~$\ct \leftarrow \Enc(\sk, m)$.
Apply $\Eval^C(\ct,\cdot) \ot \Id_B$ to $\ket\psi_{AB}$ to obtain a ciphertext~$\ct'$.
Decrypt using $m' \leftarrow \Dec(\sk,\ct')$.
Return $m'$ and register~$B$.
\end{itemize}
\item \emph{Security against quantum distinguishers:}
For any QPT algorithm~$\mcA=\{\mcA_\lambda\}$ and any two messages~$m,m' \in \M$, there is a negligible function~$\eta$ such that
\begin{align*}
    \Biggl|
    &\Pr\mleft[ 1 \leftarrow \mcA_\lambda(\ct)^{\Enc(\sk,\cdot)} \ \,\middle\vert \begin{array}{l}\sk \leftarrow \Gen(1^\lambda) \\ \ct \leftarrow \Enc(\sk, m)\end{array} \mright] \\
    &\qquad-
    \Pr\mleft[ 1 \leftarrow \mcA_\lambda(\ct)^{\Enc(\sk,\cdot)} \ \,\middle\vert \begin{array}{l}\sk \leftarrow \Gen(1^\lambda) \\ \ct \leftarrow \Enc(\sk, m')\end{array} \mright]
    \Biggr|
    \leq \eta(\lambda)
\end{align*}
for all $\lambda$.
\end{itemize}
\end{definition}

It follows from~\cite{klvy,natarajan2023bounding} that the quantum fully homomorphic encryption schemes of~\cite{qfhe,brakerskiqfhe} can be used to define QHE schemes in the sense of the above definition (note that we require correctness only for a single~$C\in\mcC$ at a time, as the security parameter tends to infinity).

We allow all subroutines to be QPT even if they only have classical input and output.
This is not important for our result and only makes it stronger, since we prove a bound that applies to any such scheme.
We also note that while the $\Eval^C$ algorithms and the correctness with auxiliary input property are required to describe the KLVY compiler and prove its correctness, they have no relevance to our result.

The security property demands that no adversary described by a QPT algorithm%
\footnote{QPT algorithms as defined in \cref{sec:prelim} are a uniform notion.
A stronger requirement is security against \emph{non-uniform} QPT quantum adversaries, in which case one can also hope to get stronger conclusions.
This is indeed the case and we return to this point in \cref{rem:non-uniform} below.}
can distinguish between the encryption of any two fixed messages, with non-negligible probability, even when given access to an encryption oracle.
This in fact implies a (seemingly stronger) security property, called \emph{parallel repeated IND-CPA security}, where the adversary can choose the two messages and also receives a polynomial number of ciphertexts~\cite{natarajan2023bounding}.

\subsection{The KLVY compiler}
We now describe the compiler of \cite{klvy}.
It takes as its input a nonlocal game (\cref{def:nonlocal_game}) and a QHE scheme (\cref{def:qhe}).
We assume from here onwards that the question and answer sets of the game are encoded as bitstrings of some fixed length.

\begin{definition}[\cite{klvy}]\label{def:comp_game}
Consider a nonlocal game~$\mcG=(\mcI_A,\mcI_B,\mcO_A,\mcO_B,\mu,V)$ and a quantum homomorphic encryption scheme~$\QHE=(\Gen,\Enc,\{\Eval^C\}_{C\in\mcC},\Dec)$ with message set~$\M \supseteq \mcI_A \cup \mcO_A$.
The corresponding \emph{compiled nonlocal game}~$\mcGcomp$ describes an interactive protocol between a verifier and a prover exchanging classical messages.
They get as input the security parameter, encoded in unary, and proceed as follows:
\begin{enumerate}[1.]
\item The verifier samples a question pair $(x,y) \leftarrow \mu$ and a secret key $\sk \leftarrow \Gen(1^\lambda)$.
They encrypt Alice's question by~$\xi \leftarrow \Enc(\sk,x)$ and send the classical ciphertext~$\xi$ to the prover.
\item The prover replies with some classical message~$\alpha$.
\item The verifier sends~$y$ unencrypted to prover.
\item The prover replies with another classical message~$b$.
\item The verifier interprets~$\alpha$ as a ciphertext and decrypts it as~$a \leftarrow \Dec(\sk,\alpha)$.
They accept if~$a\in\mcO_A$, $b \in \mcO_B$, and $V(a,b|x,y)=1$.
\end{enumerate}
\end{definition}

We only described the compiled version of a two-player nonlocal game, which is the focus of the present work, but the compiler generalizes straightforwardly to any game with $k$~players (in which case $2k$ rounds of communication are required)~\cite{klvy}.

In the compiled game, the verifier plays the role of the referee and the prover plays the role of both Alice and Bob.
In analogy to the nonlocal game, we will denote by~$\omega_\lambda(\mcGcomp,S)$ the probability that the verifier accepts for a given value of the security parameter~$\lambda\in\N$ when interacting with a prover described by a strategy~$S=\{S_\lambda\}$, where~$S_\lambda$ denotes the strategy for fixed~$\lambda$.
Using~\eqref{eq:omega G S and p}, this can also be written as
\begin{equation}\label{eq:omega lambda}
    \omega_\lambda(\mcGcomp,S)
=   \omega(\mcG,p_\lambda)
=   \hspace{-.4cm}\sum_{x \in \mcI_A, y \in \mcI_B}\sum_{a \in \mcO_A, b \in\mcO_B}\hspace{-.4cm} \mu(x,y)V(a,b|x,y)p_\lambda(a,b|x,y),
\end{equation}
where~$p_\lambda(a,b|x,y)$ denotes the probability that the prover's first reply under~$S_\lambda$ decrypts to~$a$ and that their second reply is~$b$, conditional on question pair~$(x,y)$.

Since a single prover plays the role of both Alice and Bob, this appears to be in stark contraction to the no-communication requirement of nonlocal games.
The intuition that the compiled game can still be meaningful is as follows: because we use encryption for Alice's part but not for Bob's, the prover should not be able to usefully ``correlate'' the two messages, and might therefore be forced to act like a pair of non-communicating players.
Because the security of the cryptographic scheme only applies to efficient adversaries, this can only be true if we similarly constrain the prover's computational power.

Just like in the case of the nonlocal games, there are different scenarios, depending on whether we consider classical or quantum provers.
Here we focus on the quantum scenario, since \cite{klvy} already proved that no classical efficient prover can exceed the classical value of the nonlocal game.
The following definition describes the behavior of an efficient quantum prover in the compiled setting, analogously to \cref{def:q strat} in the nonlocal setting.

\begin{definition}\label{def:qpt prover}
A \emph{QPT strategy}~$S=\{(V_\lambda,W_\lambda)\}_{\lambda \in \N}$ for a compiled game~$\mcGcomp$ consists of two QPT algorithms~$\{V_\lambda\}_{\lambda\in\N}$ and~$\{W_\lambda\}_{\lambda\in\N}$.
It describes a quantum prover that behaves as follows:
\begin{enumerate}[1.]
\item When receiving the ciphertext~$\xi\in\CT$, the prover applies~$V_\lambda$ to~$\ket\xi$ along with a suitable number of $\ket0$ states.
They then measure a suitable number of qubits and respond with the measurement outcome~$\alpha$.
\item When receiving the question~$y\in\mcI_B$, the prover applies~$W_\lambda$ to~$\ket y$ along with the post-measurement state of the preceding step.
They again measure a suitable number of qubits and respond with the measurement outcome~$b$.
\end{enumerate}
\end{definition}

This definition is perhaps more precise, but also more cumbersome to work with than the notation used in prior works, which instead described QPT strategies by families~$\{(\mcH_\lambda,\ket{\psi_\lambda},\{A^\lambda_{\xi\alpha}\},\{B^\lambda_{yb}\})\}_{\lambda\in\N}$, consisting of
\begin{itemize}
    \item Hilbert spaces~$\mcH_\lambda$,
    \item states~$\ket{\psi_\lambda} \in \mcH_\lambda$,
    \item measurement operators of the form~$A^\lambda_{\xi\alpha} = U^\lambda_{\xi\alpha} P^\lambda_{\xi\alpha}$, where all~$U^\lambda_{\xi\alpha}$ are unitaries on~$\mcH_\lambda$ and $\{P^\lambda_{\xi\alpha}\}_{\alpha\in\CT}$ is a projective measurement for any~$\xi\in\CT$,
    \item POVMs or projective measurements~$\{B^\lambda_{yb}\}_{b\in\mcO_B}$ for each~$y\in\mcI_B$,
\end{itemize}
subject to QPT assumptions that are less straightforward to state than above.
The relation is immediate: we take $\ket{\psi_\lambda}$ to be the all-zeros state on a suitable multi-qubit Hilbert space~$\mcH_\lambda$ (but it can be any state that can be prepared by a QPT algorithm), the projective measurements~$P^\lambda_{\xi\alpha}$ correspond to the first part of \cref{def:qpt prover}, the unitaries~$U^\lambda_{\xi\alpha}$ can be taken as the identity,%
\footnote{This is without loss of generality: the unitaries~$U^\lambda_{\xi\alpha}$ can always be absorbed into the second POVM, as~$\ket\alpha$ is part of the post-measurement state and we can also keep a copy of~$\ket\xi$ in it.}
and the operators~$B^\lambda_{yb}$ correspond to the second part of \cref{def:qpt prover}, that is,
\begin{align}\label{eq:def B lambda}
    B^\lambda_{yb} = (\bra y \ot \Id) W_\lambda^* (\ketbra b b \ot \Id) W_\lambda (\ket y \ot \Id).
\end{align}
We emphasize that unlike in the nonlocal case, there are \emph{no} commutation conditions imposed on the operators~$A_{\xi\alpha}^\lambda$ and~$B_{yb}^\lambda$, nor is there any tensor product structure of the Hilbert spaces~$\mcH_\lambda$.
Using this notation, the probabilities $p_\lambda$ in~\eqref{eq:omega lambda} take the following form for a QPT strategy:
\begin{align*}
    p_\lambda(a,b|x,y)
&= \hspace{-0.3cm}\expect_{\sk\leftarrow\Gen(1^\lambda)} \expect_{\xi\leftarrow\Enc(\sk,x)} \sum_{\alpha\in\CT} \Pr(a\leftarrow\Dec(\sk,\alpha))
\, \bra{\psi^\lambda} \, (A_{\xi\alpha}^\lambda)^*B_{yb}^\lambda A_{\xi\alpha}^\lambda\ket{\psi^\lambda}
\end{align*}
Introducing positive (semidefinite) operators
\begin{align}\label{eq:sigma_xa}
  \sigma_{xa}^\lambda =
    \hspace{-0.3cm}\expect_{\sk\leftarrow\Gen(1^\lambda)} \expect_{\xi\leftarrow\Enc(\sk,x)} \sum_{\alpha\in\CT} \Pr(a\leftarrow\Dec(\sk,\alpha))
  \, A_{\xi\alpha}^\lambda\ketbra{\psi^\lambda}{\psi^\lambda} (A_{\xi\alpha}^\lambda)^*,
\end{align}
the correlations can also be written as
\begin{align}\label{eq:p lambda qpt prover}
    p_\lambda(a,b|x,y)
&= \tr(\sigma^\lambda_{xa} B_{yb}^\lambda).
\end{align}
We note that $\tr(\sigma_{xa}^\lambda)$ is the probability that the first part of the prover replies with a ciphertext that decrypts to~$a\in\mcO_A$ when given an encryption of~$x\in\mcI_A$, and~$\sigma_{xa}^\lambda / \tr(\sigma_{xa}^\lambda)$ is its post-measurement state in this case.

As in the nonlocal case, we can also define the value of a compiled game by taking the maximum (supremum) over all possible QPT strategies.
Any quantum strategy~$S = \{S_\lambda\}$ gives rise to a \emph{sequence} of acceptance probabilities~$\{\omega_\lambda(\mcGcomp,S)\}$ as in \cref{eq:omega lambda}.
If we would like to define a single value then there are at least two natural definitions.

\begin{definition}\label{def:q values compiled game}
Then we associate to a compiled game~$\mcGcomp$ the following \emph{minimal and maximal quantum values}:
\begin{align*}
    \omega_{q,\min}(\mcGcomp) &= \sup \braces*{ \liminf_{\lambda\to\infty} \omega_\lambda(\mcGcomp,S) \ \middle|\ S=\{S_\lambda\} \text{ a QPT strategy for } \mcGcomp }, \\
    \omega_{q,\max}(\mcGcomp) &= \sup \braces*{ \limsup_{\lambda\to\infty} \omega_\lambda(\mcGcomp,S) \ \middle|\ S=\{S_\lambda\} \text{ a QPT strategy for } \mcGcomp }.
\end{align*}
\end{definition}

Both quantities are meaningful.
A bound of the form~$\omega_{q,\min}(\mcGcomp) \geq \theta$ shows that efficient quantum provers are able to \emph{achieve} an acceptance probability arbitrarily close to~$\theta$, while $\omega_{q,\max}(\mcGcomp,S) \leq \theta$ means that no quantum prover can \emph{exceed} this acceptance probability by any constant, for large enough security parameter.
Clearly, $\omega_{q,\min}(\mcGcomp) \leq \omega_{q,\max}(\mcGcomp)$.

Any strategy for the nonlocal game can be converted into a prover for the compiled game by using the homomorphic evaluation functionality of the encryption (assuming it supports evaluating the necessary quantum circuits).
Thus, for every quantum strategy~$S$ for the nonlocal game~$\mcG$, there exists a QPT strategy~$\Scomp=\{\Scomp^\lambda\}$ and a negligible function~$\eta$ such that~$\omega_\lambda(\mcGcomp,\Scomp) \geq \omega(\mcG,S) - \eta(\lambda)$ for all~$\lambda\in\N$.
This is one half of the main result~\cite[Thm.~3.2]{klvy}, and it implies that in particular~$\omega_{q,\min}(\mcGcomp) \geq \omega_q(\mcG)$.
The other half of their theorem states that efficient classical provers cannot exceed the classical value~$\omega_c(\mcG)$ by a non-negligible amount, as already mentioned earlier.

\begin{remark}\label{rem:non-uniform}
In this section, we consider quantum strategies that are defined in terms of (uniform) QPT algorithms, in line with the prior works~\cite{natarajan2023bounding,xor} and to emphasize that all reductions that will be discussed in the following are uniform as well.
In cryptography, one can also model adversaries by \emph{non-uniform} QPT algorithms, as mentioned earlier, and one can similarly define non-uniform QPT strategies.
It is easy to see that all our results hold verbatim for such strategies, provided the QHE scheme is secure against non-uniform QPT adversaries.

We note that while in this setting the appropriate definition of~$\omega_{q,\max}$ is by optimizing over non-uniform QPT strategies,
$\omega_{q,\min}$ is still most naturally defined in terms of uniform QPT strategies since this is the appropriate notion for an honest prover to achieve a desired functionality.
\end{remark}

\subsection{Asymptotic security for any noncommutative polynomial}
In the compiled game the prover gets Bob's question after giving Alice's answer.
This implies that the correlations~$p_\lambda(a,b|x,y)$ are necessarily non-signaling \emph{from Bob to Alice}, i.e.\ $\sum_{b \in \mcO_B} p_\lambda(a,b|x,y) = \sum_{b \in \mcO_B} p_\lambda(a,b|x,y')$ for all $x\in\mcI_A$ and~$y,y'\in\mcI_B$.
Unlike in the nonlocal case, however, it is not a priori clear to which extent these correlations are non-signaling \emph{from Alice to Bob}.
However, the security property of the QHE scheme readily implies that the post-measurement states~\eqref{eq:sigma_xa} are computationally indistinguishable when averaged over the possible measurement outcomes~$\alpha$.
That is, if we define the quantum states
\begin{equation}\label{eq:sigma_x}
    \sigma^\lambda_x
= \sum_{a\in\mcO_A} \sigma^\lambda_{xa}
= \hspace{-0.3cm}\expect_{\sk\leftarrow\Gen(1^\lambda)} \expect_{\xi\leftarrow\Enc(\sk,x)} \sum_{\alpha\in\CT}
  A_{\xi\alpha}^\lambda\ketbra{\psi^\lambda}{\psi^\lambda} (A_{\xi\alpha}^\lambda)^*.
\end{equation}
then~$\{\sigma^\lambda_x\}$ and $\{\sigma^\lambda_{x'}\}$ are computationally indistinguishable for any~$x,x'\in\mcI_A$, meaning that no QPT algorithm can distinguish them with non-negligible probability.
This, in particular, implies that the correlations~$p_\lambda(a,b|x,y)$ become non-signaling from Alice to Bob in the limit~$\lambda\to\infty$.
That is, for any~$x,x'\in\mcI_A$ and~$y\in\mcI_B$, there exists a negligible function~$\eta$ such that, for all~$\lambda$,
\begin{align*}
    \abs[\Big]{ \sum_{a\in\mcO_A} p_\lambda(a,b|x,y) - \sum_{a\in\mcO_A} p_\lambda(a,b|x',y) }
=  \abs[\Big]{ \tr(\sigma^\lambda_x B^\lambda_{yb}) - \tr(\sigma^\lambda_{x'} B^\lambda_{yb}) }
\leq \eta(\lambda),
\end{align*}
because the POVM~$\{B^\lambda_{yb}\}$ is implemented by a QPT algorithm for every~$y \in \mcI_B$.

However, the security property of the encryption scheme implies a much stronger notion of computational non-signaling from Alice to Bob, as it makes a statement about any efficient algorithm.
In particular, we can prove the following result.

\begin{proposition}\label{prop:polysecurity}
Consider any compiled game and QPT strategy.
Let $x,x'\in\mcI_A$, and let~$P=P(\{B_{yb}\})$ be a polynomial in noncommuting variables~$\{B_{yb}\}_{y\in\mcI_B,b\in\mcO_B}$.
Then there exists a negligible function~$\eta$ such that, for all $\lambda\in\N$,
\begin{align*}
    \abs*{ \tr\mleft( \sigma^\lambda_x \, P(\{B^\lambda_{yb}\}) \mright) - \tr\mleft( \sigma^\lambda_{x'} \, P(\{B^\lambda_{yb}\}) \mright) }
\leq \eta(\lambda),
\end{align*}
where~$\sigma^\lambda_x$ is the prover's average state after its first reply when given an encryption of~$x\in\mcI_A$, see~\eqref{eq:sigma_x}, and where~$\{B^\lambda_{yb}\}_{b \in \mcO_B}$ are POVMs for~$y\in\mcI_B$, corresponding to the measurements that lead to the prover's second reply, as defined in~\eqref{eq:def B lambda}.
\end{proposition}

\Cref{prop:polysecurity} is a generalization of \cite[Lem.~21]{xor} and is proved in a similar fashion, by using block encodings.

\begin{definition}
    A \textit{block encoding} of an operator~$M$ on $(\C^2)^{\ot n}$ is a unitary~$U$ on~$(\C^2)^{\ot (m+n)}$, for some additional number of qubits~$m\in \N$, such that
    \[
        tM = \left( \bra{0}^{\ot m} \ot \Id\right) U \left(\ket{0}^{\ot m} \ot \Id\right), \quad \text{i.e.} \quad U = \begin{pmatrix} tM & * \\ * & * \end{pmatrix}
    \]
    for some $t > 0$ called the \textit{scale factor} of the block encoding.

    A \emph{QPT block encoding} of a family of operators~$\{M_\lambda\}$ is a QPT algorithm~$\{U_\lambda\}$ such that each~$U_\lambda$ is a block encoding of~$M_\lambda$, with~$t$ and~$m$ independent of~$\lambda$.
\end{definition}

The significance of this definition is as follows.
On the one hand, the quantum expectation value of an observable that admits a QPT block encoding can be measured to any inverse polynomial precision, by a QPT quantum algorithm that takes polynomially many copies of the state.
Together with the security property of the QHE scheme this implies the following.

\begin{lemma}[{\cite[Lem.~2.21]{xor}, cf.\ \cite[Lem.~15-17]{natarajan2023bounding}}]\label{lem:block ind}
Consider any compiled game and QPT strategy.
Let $x,x'\in\mcI_A$, and let~$\{M_\lambda\}_{\lambda \in \N}$ be a family of observables that admit a QPT block encoding and such that~$\sup_\lambda \,\norm{M_\lambda} < \infty$.
Then there exists a negligible function~$\eta$ such that, for all $\lambda\in\N$,
\begin{align*}
    \abs*{ \tr\mleft( \sigma^\lambda_x \, M_\lambda \mright) - \tr\mleft( \sigma^\lambda_{x'} \, M_\lambda \mright) }
\leq \eta(\lambda),
\end{align*}
where the states~$\sigma^\lambda_x$ are defined as in~\eqref{eq:sigma_x}.
\end{lemma}

On the other hand, one can show that the families of POVM elements~$\{B^\lambda_{yb}\}_{\lambda\in\N}$ have natural QPT block encodings, and moreover that the existence of QPT block encoding is preserved by multiplication and taking linear combinations.
This implies that the family of operators defined by $M_\lambda = P(\{B^\lambda_{yb}\})$ has a QPT block encoding, which in view of the preceding lemma essentially establishes the proposition.
The following proof makes this reasoning precise.

\begin{proof}[Proof of \cref{prop:polysecurity}]
It suffices to prove the claim for monomials since any noncommutative polynomial is a finite linear combination of monomials.

We first note that for any fixed~$y\in\mcI_B$ and~$b\in\mcO_B$, the POVM elements~$\{B^\lambda_{yb}\}$ have a natural QPT block encoding.
This follows from equation \eqref{eq:def B lambda}, by the same reasoning as in~\cite[Lem.~26]{gilyen2019quantum}, which also shows that the resulting block encodings have parameter~$t=1$ and~$m=N$, where~$N$ denotes the total number of bits in the binary representation of~$y$ and~$b$.

Now suppose that~$P$ is a monomial of degree~$D$ and let~$M_\lambda = P(\{B^\lambda_{yb}\})$.
It follows from~\cite[Lem.~30]{gilyen2019quantum} that~$\{M_\lambda\}$ admits a QPT block encoding with~$t=1$ and~$m=DN$, simply by concatenating the individual block encodings in a suitable way (cf.\ \cite[Lem~2.18]{xor}).
Even though each POVM element is an observable, the operators~$M_\lambda$ need not be Hermitian, so we cannot apply \cref{lem:block ind} directly.
Instead, we observe that it suffices to prove the claim for the observables~$\Re(M_\lambda) = (M_\lambda + M_\lambda^*)/2$ and $\Im(M_\lambda) = (M_\lambda - M_\lambda^*)/(2 \imath)$.
To this end, we first note that~$\{M_\lambda^*\}$ also admits a QPT block encoding with the same parameters as~$\{M_\lambda\}$, as the former family corresponds to the monomial obtained by reversing~$P$.
Then it follows from~\cite[Lem.~29]{gilyen2019quantum}, by taking the controlled unitaries corresponding to these QPT block encodings, along with fixed state-preparation pairs for the two desired linear combinations, that the two families~$\{\Re(M_\lambda)\}$ and~$\{\Im(M_\lambda)\}$ admit QPT block encodings (cf.\ \cite[Lem~2.17]{xor}).
Now the claim follows from \cref{lem:block ind}, the triangle inequality, and the fact that nonnegative linear combinations of negligible functions are negligible.
\end{proof}

\section{Sequential characterizations of nonlocal correlations}\label{sec:alternative-characterization}
Motivated by the two-round structure of a compiled game, we consider sequential games and strategies.
Without further constraints, the resulting correlations can even be signaling, but we find that a natural information-theoretic property motivated by \cref{prop:polysecurity} ensures that the resulting correlations are nonlocal ones.
We describe the sequential setting in \cref{subsec:seq games}.
In \cref{subsec:seq classical,subsec:seq quantum} we discuss how to characterize classical as well as (finite-dimensional and spatial) quantum correlations in this setting.
In \cref{subsec:seq com} we prove the main result of this section: a ``sequential'' characterization of commuting operator correlations (\cref{thm:quantum strong nonsig c*}), which to the best of our knowledge has not appeared in the literature before.

Only the latter is required for the main results of this article, and we invite readers interested only in those to proceed directly to the self-contained \cref{subsec:seq com}.
The pedagogical \cref{subsec:seq classical,subsec:seq quantum} discuss more concrete settings that allow to gain intuition for the more abstract algebraic results of \cref{subsec:seq com}, and the results therein can be obtained from prior work, as we explain below.

\subsection{Sequential games}\label{subsec:seq games}
We consider sequential games that are parameterized by nonlocal games (\cref{def:nonlocal_game}).
Unlike in the nonlocal game, there is a single player that plays the roles of both Alice and Bob.

\begin{definition}
Consider a nonlocal game~$\mcG=(\mcI_A,\mcI_B,\mcO_A,\mcO_B,\mu,V)$.
The corresponding \emph{sequential game}~$\mcGseq$ describes a scenario of a single player interacting with a referee.
In the game, the referee samples a pair of questions~$(x,y) \in \mcI_A \times \mcI_B$ according to $\mu$ and sends question~$x$ to the player.
The player returns an answer~$a\in\mcO_A$.
Then the referee sends~$y$ to the player, who replies with an answer~$b\in\mcO_B$.
Finally the referee computes~$V(a,b|x,y)$ to determine if the player wins or loses.
\end{definition}

\begin{remark}
A sequential game can also be interpreted as a \emph{two-player} game where the first player can pass some information (depending on their question) to the second player before the latter has to respond with their answer.
Sequential games have been investigated from different perspectives; see, e.g., \cite{catani2024connecting} and the references therein.

We note that this setting can also be translated into the language of \emph{prepare-and-measure scenarios} studied in the contextuality literature.
In particular, \cite{wright2023invertible} proposes a general map between nonlocal scenarios and strategies on the one hand and certain prepare-and-measure contextuality scenarios and preparations on the other hand (their introduction also gives an account of the motivations of this line of research and prior works).
\Cref{cor:seq classical,cor:seq quantum} can be obtained from their work (which in turn builds on \cite{NCPV12}), while our \cref{thm:quantum strong nonsig c*,cor:eqcharacterization} resolve an open question left in their work.
\end{remark}

As in the nonlocal case, we can describe the player's behavior by strategies that determine the probabilities~$p(a,b|x,y)$ of answers~$a,b$ given questions~$x,y$.
Thus, the probability of winning the game~$\mcGseq$ under a sequential strategy~$S$, with correlations~$p=\{p(a,b|x,y)\}$, will be denoted by
\begin{equation*}
    \omega(\mcGseq,S)
=   \omega(\mcG,p)
=   \sum_{x \in \mcI_A, y \in \mcI_B}\sum_{a \in \mcO_A, b \in\mcO_B}\mu(x,y)V(a,b|x,y)p(a,b|x,y).
\end{equation*}
Because of the temporal order in the sequential game, these correlations should be non-signaling from Bob to Alice, meaning that for all $x \in \mcI_A$, $y,y'\in \mcI_B$, and $a \in \mcO_A$, it should hold that
\begin{align*}
    \sum_{b\in \mcO_B}p(a,b|x,y)&=\sum_{b\in \mcO_B}p(a,b|x,y').
\end{align*}
On the other hand, there is nothing imposed that prevents Alice from signaling Bob.

\subsection{Classical strategies and correlations}\label{subsec:seq classical}
While our main interest is quantum strategies, we first discuss the classical case to build some intuition.

\begin{definition}\label{def:c strat seq}
A \emph{classical strategy} for the sequential game~$\mcGseq$ consists~of
\begin{enumerate}[(i)]
    \item probability distributions $\{p(a,\omega|x) : x\in \mcI_A \}$ with outcomes in~$\mcO_A \times \Omega$, where~$\Omega$ is a (without loss of generality) finite set,
    \item probability distributions $\{q_\omega(b|y) : y\in \mcI_B,\omega\in \Omega\}$ with outcomes in~$\mcO_B$.
\end{enumerate}
Such a classical strategy gives rise to a correlation
\begin{equation*}
    p(a,b|x,y)
=   \sum_{\omega\in \Omega} p(a,\omega|x) \, q_\omega(b|y),
\end{equation*}
where $a\in \mcO_A,b\in \mcO_B,x\in \mcI_A,y\in \mcI_B$.
\end{definition}

We note that~$\omega\in\Omega$ models the information that is preserved between the two rounds of the game.
While classical strategies for sequential games are always non-signaling from Bob to Alice, they may even be signaling from Alice to Bob.
However, we can identify a natural property that ensures that the resulting correlations are not only non-signaling, but in fact nonlocal classical correlations in the sense of \cref{def:c strat}:
the distribution of~$\omega$ should be independent of~$x$.

\begin{proposition}\label{prop:classical strong nonsig}
Consider a classical strategy for~$\mcGseq$ and suppose that the distributions~$p(\omega|x) = \sum_{a\in\mcO_A} p(a,\omega|x)$ are the same for all~$x\in\mcI_A$.
Then the resulting correlation is a (nonlocal) classical correlation, that is, in~$C_c$.
\end{proposition}
\begin{proof}
Define a probability distribution $\gamma(\omega) := p(\omega|x)$, which by assumption does not depend on~$x\in\mcI_A$, as well as probability distribution~$p_\omega(a|x) := p(a,\omega|x) / \gamma(\omega)$ for~$x\in\mcI_A$ and~$\omega\in\Omega$ (for $\gamma(\omega)=0$ the corresponding distributions~$p_\omega$ can be defined arbitrarily).
Then it holds that
\begin{align*}
    p(a,b|x,y)
=   \sum_{\omega\in \Omega} p(a,\omega|x) \, q_\omega(b|y)
=   \sum_{\omega\in \Omega} \gamma(\omega) \, p_\omega(a|x) \, q_\omega(b|y),
\end{align*}
which is precisely the form of a classical correlation.
\end{proof}

Conversely, any classical strategy for the nonlocal game (\cref{def:c strat}) gives rise to one for the sequential game that satisfies the hypotheses of \cref{prop:classical strong nonsig}.
Simply set~$p(a,\omega|x) := \gamma(\omega) p_\omega(a|x)$ and use the same~$q_\omega(b|y)$ as in the~nonlocal~strategy.
We thus obtain the following characterization, which can also be obtained by translating~\cite[App.~D]{wright2023invertible} into the language of sequential games:

\begin{corollary}\label{cor:seq classical}
The classical correlation set~$C_c$ consists precisely of the correlations produced by classical sequential strategies satisfying the condition in \cref{prop:classical strong nonsig}.
\end{corollary}

\begin{remark}\label{rem:sigma algebra}
The condition identified in \cref{prop:classical strong nonsig} does not refer to the Bob part of the strategy.
If one takes Bob's strategy into account then one can give a sharper criterion -- the distributions~$p(\omega|x)$ should coincide when restricted to the $\sigma$-algebra generated by the functions~$\{\omega\mapsto q_\omega(b|y)\}_{b\in\mcO_B,y\in\mcI_B}$.
In the classical case, the simpler condition is without loss of generality, but not so in the (infinite-dimensional) quantum case.
\end{remark}

\subsection{Quantum strategies and correlations}\label{subsec:seq quantum}
Next, we move on to the quantum case.
We give a definition that applies in finite as well as infinite dimensions.

\begin{definition}\label{def:q strat seq}
A \emph{quantum strategy} for the sequential game~$\mcGseq$ consists of
\begin{enumerate}[(i)]
\item a Hilbert space~$\mcH$,
\item\label{it:sigma data} positive (semidefinite) operators~$\{\sigma_{xa}\}_{x\in\mcI_A,a\in\mcO_A}$ such that~$\sigma_x := \sum_{a\in\mcO_A} \sigma_{xa}$ is a density operator (i.e., has unit trace) for every~$x\in\mcI_A$, along with
\item POVMs $\{\{B_{yb}:b\in \mcO_B\}: y\in \mcI_B\}$ acting on $\mcH$.
\end{enumerate}
(We say that the strategy is \emph{finite-dimensional} if the Hilbert space~$\mcH$ is finite-dimensional.)
Such a quantum strategy gives rise to a correlation
\begin{equation*}
    p(a,b|x,y)
=   \tr\mleft( \sigma_{xa} B_{yb} \mright)
\end{equation*}
where $a\in \mcO_A,b\in \mcO_B,x\in \mcI_A,y\in \mcI_B$.
\end{definition}
Note that this formula is precisely the same expression as in~\eqref{eq:p lambda qpt prover} for the correlation determined by a QPT strategy.

Operators as in~\ref{it:sigma data} naturally arise as unnormalized post-measurement states for quantum measurements.
For example, given a state~$\rho$ and a collection of measurements~$\{\{A_{xa}\}_{a\in \mcO_A}\}_{x\in \mcI}$, the operators~$\sigma_{xa} = A_{xa} \rho A_{xa}^*$ satisfy the assumption, as do the operators~$\{\sigma_{xa}^\lambda\}$ defined in~\eqref{eq:sigma_xa} for the compiled game (for any fixed~$\lambda$).
This is immediate, but can also be verified by the following lemma.

\begin{lemma}\label{lem:state vs instrument}
Let $\mcH$ be a Hilbert space, $\rho$ be a state on~$\mcH$, and~$\{\Phi_{xa}\}_{x\in\mcI_A,a\in\mcO_A}$ be a collection of completely positive maps such that~$\sum_{a\in\mcO_A} \Phi_{xa}$ is trace-preserving for every~$x\in\mcI_A$.%
\footnote{A collection of completely positive maps~$\{\Psi_a\}_{a\in\mcO}$ such that $\sum_{a\in\mcO_A} \Psi_{xa}$ is trace-preserving is called a \emph{quantum instrument}. It describes the most general quantum evolution that has a classical outcome (measurement result) as well as a quantum one (post-measurement state)~\cite[\S{}4.6.8]{wilde2017quantum}.}
Then the operators~$\sigma_{xa} = \Phi_{xa}(\rho)$ satisfy the assumptions in~\ref{it:sigma data} of \cref{def:q strat seq}.
Conversely, any collection of operators as in~\ref{it:sigma data} arises in this way.
\end{lemma}
\begin{proof}
The first claim follows directly from the trace-preserving assumption.
For the converse, take~$\rho$ to be an arbitrary state and define $\Phi_{xa}(\cdot) = \tr(\cdot) \, \sigma_{xa}$.
\end{proof}

Correlations produced by quantum strategies for sequential games are always non-signaling from Bob to Alice, but not necessarily from Alice to Bob.
We now state the key property that allows us to ensure that the correlations are in fact quantum respective commuting operator correlations for the nonlocal game, in the sense of \cref{def:q strat,def:qc_strat}.
As the states~$\sigma_x$ are analogous to the marginal distributions~$p(\omega|x)$ in the classical case, it is natural to demand that they are identical in a suitable sense.

\begin{definition}\label{def:q str nonsig}
We say that a quantum strategy is \emph{strongly non-signaling} if there exists a $C^*$-algebra $\msB \subseteq \mcB(\mcH)$ containing the operators~$\{B_{yb}\}_{y\in\mcI_B,b\in\mcO_B}$ such the following condition holds:
for all $x,x'\in\mcI_A$ and for all~$B \in \msB$, we~have
\begin{align}\label{eq:q str nonsig}
    \tr( \sigma_x B ) = \tr( \sigma_{x'} B ).
\end{align}
We say that~$S$ is \emph{strongly non-signaling with respect to~$\msB$} to indicate~$\msB$ explicitly.
\end{definition}

Note that a quantum strategy is strongly non-signaling if, and only if, it is so with respect to the $C^*$-algebra generated by the operators~$\{B_{yb}\}$ because it is contained in any other $C^*$-algebra that contains these elements.
By continuity it suffices to verify~\eqref{eq:q str nonsig} on the dense set of noncommutative polynomials in these operators.%
\footnote{However, it does \emph{not} suffice to only require that~\eqref{eq:q str nonsig} holds for the generators~$B \in \{ B_{yb} \}_{y\in\mcI_B,b\in\mcO_B}$.}
We record this useful observation.

\begin{lemma}\label{lem:q str nonsig via polys}
A quantum sequential strategy is strongly non-signaling if, and only if, for every~$x,x'\in\mcI_A$ and for every noncommutative polynomial~$P(\{B_{yb}\})$ in the operators~$\{B_{yb}\}_{y\in\mcI_B,b\in\mcO_B}$, it holds that
\begin{align}\label{eq:poly cond}
    \tr\bigl( \sigma_x \, P(\{B_{yb}\}) \bigr) = \tr\bigl( \sigma_{x'} \, P(\{B_{yb}\}) \bigr).
\end{align}
\end{lemma}

Note the similarity between the characterization in \cref{lem:q str nonsig via polys} and the statement of \cref{prop:polysecurity}.
In \cref{sec:new-magic} we will show how to connect the two in a precise~way.
Because we will need to take the limit where the security parameter tends to infinity, this will require us to consider infinite-dimensional strategies.

To build intuition we first consider the special case of quantum sequential strategies that are strongly non-signaling with respect to $\mcB(\mcH)$ for some Hilbert space $\mcH$.
In other words, we consider the situation that $\sigma_x = \sigma_{x'}$ for all~$x,x'\in\mcI_A$.
This situation was also studied in~\cite{NCPV12} in a different context.
There, the corresponding correlations were called ``quansal'' and it was proved that every \emph{quansal} correlation is in~$C_{qs}$.
We give a self-contained proof of this result:


\begin{proposition}[{\cite[Lem.~4]{NCPV12}}]\label{prop:quantum strong nonsig spatial}
Consider a quantum strategy for~$\mcGseq$ that is strongly non-signaling with respect to $\mcB(\mcH)$.
Then the resulting correlation is a (nonlocal) spatial quantum correlation, that is, in~$C_{qs}$.
\end{proposition}
\begin{proof}
Since the quantum strategy is strong non-signaling condition with respect to $\mcB(\mcH)$, it holds $\sigma_x = \sigma_{x'}$ for all~$x,x'\in\mcI_A$, and therefore we let~$\sigma=\sigma_x$. Let $\mcH_A := \C^{\mcO_A} \ot \mcH_{A'}$ and~$\mcH_{A'} := \mcH_B := \mcH$.
We think of the operators~$\sigma_{xa}$ as acting on~$\mcH_B$ and choose purifications~$\ket{\psi_{xa}} \in \mcH_{A'} \ot \mcH_B$ for each~$x\in\mcI_A$ and~$a\in\mcO_A$.
Then the states
\begin{align*}
    \ket{\psi_x} := \sum_{a\in\mcO_A} \ket a \ot \ket{\psi_{xa}} \in \mcH_A \ot \mcH_B
\end{align*}
are purifications of the same operator~$\sigma$, which implies that there exist unitaries~$U_{xx'}$ on $\mcH_A$ such that
$\ket{\psi_x} = (U_{xx'} \ot \Id_B) \ket{\psi_{x'}}$.
Defining $P_a := \ketbra a a \ot \Id_{A'}$ gives a projective measurement~$\{P_a\}_{a\in\mcO_A}$ on~$\mcH_A$.
Fixing some~$x_0\in\mcI_A$, we observe that
\begin{align*}
    p(a,b|x,y)
&=   \tr\mleft( \sigma_{xa} B_{yb} \mright)
=   \bra{\psi_x} P_a \ot B_{yb} \ket{\psi_x}
=   \bra{\psi_{x_0}} U_{xx_0}^* P_a U_{xx_0} \ot B_{yb} \ket{\psi_{x_0}},
\end{align*}
which shows that~$p(a,b|x,y)$ is a quantum spatial correlation, with Hilbert spaces~$\mcH_A$, $\mcH_B$, initial state~$\ket{\psi_{x_0}}$, and POVM elements~$A_{xa} := U_{xx_0}^* P_a U_{xx_0}$ and~$B_{yb}$ for all $a\in \mcO_A, b\in \mcO_B, x\in \mcI_A$, and $y\in \mcI_B$.
\end{proof}

It is clear by inspection of the preceding proof that if~$\mcH$ is finite-dimensional then the resulting correlation is a quantum correlation, that is, in~$C_q$.
In fact, in the finite-dimensional case, we do not have to assume that~$\msB = \mcB(\mcH)$:

\begin{proposition}\label{prop:fin dim seq}
Consider a finite-dimensional quantum strategy for~$\mcGseq$ that is strongly non-signaling.
Then the resulting correlation is a (nonlocal) quantum correlation, that is, in~$C_q$.
\end{proposition}
\begin{proof}
By the preceding discussion, it suffices to show that for every such strategy one can find some other strategy that produces the same correlations but is strongly non-signaling for~$\mcB(\mcH)$.
Then the result follows from \cref{prop:quantum strong nonsig spatial}.
By the classification of finite-dimensional $*$-algebras, we may assume that
\begin{align*}
    \mcH = \bigoplus_i (\C^{n_i} \otimes \C^{m_i}), \;\;
    \msB = \bigoplus_i M_{n_i}(\C)\otimes \Id_{m_i}, \;\;
    \text{and} \;\; \msB'=\bigoplus_i \Id_{n_i}\otimes M_{m_i}(\C),
\end{align*}
where $\msB'$ denotes the commutant of $\msB$.
Denote by $P_i$ the projection onto the $i$th direct summand and define positive operators~$\tilde{\sigma}_{xa}=\bigoplus_i \tr_{M_{m_i}(\C)}(P_i \sigma_{xa} P_i)\otimes \frac{\Id_{m_i}}{m_i}$ for~$x\in\mcI_A$ and~$a\in\mcO_A$, where $\tr_{M_{m_i}(\C)}$ is the partial trace over the second tensor factor of the $i$th summand.
Let $B=\bigoplus_i B^{(i)} \otimes \Id_{m_i}$ be an arbitrary element in~$\msB$.
We compute: 
\begin{align*}
   \tr(\sigma_{xa} B)
&= \sum_i \tr\bigl( \sigma_{xa} P_i (B^{(i)} \otimes \Id_{m_i}) P_i \bigr)\\
&= \sum_i \tr\bigl( P_i \sigma_{xa} P_i (B^{(i)} \otimes \Id_{m_i}) \bigr)\\
&= \sum_i \tr\bigl(\tr_{M_{m_i}(\C)}(P_i\sigma_{xa}P_i)B^{(i)}\bigr)\\
&= \sum_i \tr\bigl( (\tr_{M_{m_i}(\C)}(P_i \sigma_{xa} P_i)\otimes \tfrac{\Id_{m_i}}{m_i}) (B^{(i)} \otimes \Id_{m_i}) \bigr)\\
&= \tr(\tilde{\sigma}_{xa}B).
\end{align*}
On the one hand, this shows that the operators~$\tilde\sigma_{xa}$ produce the same correlations as the operators~$\sigma_{xa}$.
On the other hand, it follows that $\tr(\tilde{\sigma}_{x} B)=\tr(\tilde{\sigma}_{x'} B)$ for every $B\in \msB$, because we know that $\tr(\sigma_{x} B)=\tr(\sigma_{x'} B)$ by assumption.
Because the operators~$\tilde{\sigma}_x,\tilde{\sigma}_{x'}$ are themselves elements in~$\msB$, it follows that~$\tilde{\sigma}_{x} = \tilde{\sigma}_{x'}$ for all~$x,x'\in\mcI_A$.
Thus, $\tilde{\sigma}_{xa}$ is strongly non-signaling with respect to~$\mcB(\mcH)$ and produces the same correlations as the original strategy.
\end{proof}

Conversely, any quantum (spatial) strategy for the nonlocal game gives rise to one for the sequential game that satisfies the hypotheses of \cref{prop:quantum strong nonsig spatial}.
Simply let~$\mcH = \mcH_B$, $\sigma_{xa} = \tr_A( (M_{xa} \ot \Id_B) \ketbra\psi \psi )$, and use the operators~$B_{yb} = N_{yb}$.
Then strong non-signaling is satisfied with~$\msB = \mcB(\mcH)$.
We summarize:

\begin{corollary}\label{cor:seq quantum}~
\begin{itemize}
    \item[(i)] The quantum spatial correlation set~$C_{qs}$ consists precisely of the correlations produced by quantum sequential strategies that are strongly non-signaling \emph{with respect to $\mcB(\mcH)$}.
    \item[(ii)] The quantum correlation set~$C_q$ consists precisely of the correlations produced by \emph{finite-dimensional} quantum sequential strategies that are strongly non-signaling (we can but need not assume that~$\msB=\mcB(\mcH$)).
\end{itemize}
\end{corollary}

The preceding \cref{cor:seq quantum} only characterizes the correlations produced by strongly non-signaling quantum correlations in special cases.
Indeed, \cref{def:q str nonsig} only requires equality on the subalgebra generated by Bob's measurement operators.
As such, the preceding results are instructive for building intuition but are insufficient for our main application, which requires the general infinite-dimensional case.
Here, unlike in finite dimensions, it is not possible to reduce to the case that~$\msB=\mcB(\mcH)$, and indeed we will find that the strongly non-signaling condition corresponds to general commuting operator correlations, as we prove in the next section.

\subsection{A characterization of quantum commuting correlations}\label{subsec:seq com}
In this subsection, we find that the strong non-signaling condition precisely characterizes the \emph{commuting operator correlations} in the sense of \cref{def:qc_strat}.
To establish this result, it is useful to define the following equivalent $C^*$-algebraic model.

\begin{definition}\label{def:nonsig a strat seq}
A \emph{strongly non-signaling algebraic strategy} consists of
\begin{enumerate}[(i)]
\item a $C^*$-algebra $\msB$,
\item positive linear functionals $\phi_{xa}\colon\msB\to\C$ for~$x\in\mcI_A$ and~$a\in\mcO_A$, along with
\item POVMs $\{B_{yb}\}_{b\in\mcO_B}$ in~$\msB$ with outcomes in~$\mcO_B$ for every~$y\in\mcI_B$,
\end{enumerate}
such that there exists a state~$\phi\colon\msB\to\C$ such that $\sum_{a\in\mcO_A} \phi_{xa} = \phi$ for every~$x\in\mcI_A$.
Such a strategy gives rise to a correlation
\begin{align*}
    p(a,b|x,y) = \phi_{xa}(B_{yb})
\end{align*}
where $a\in \mcO_A,b\in \mcO_B,x\in \mcI_A,y\in \mcI_B$.
\end{definition}

Any quantum strategy for~$\mcG$ can be converted into such an algebraic strategy
provided it satisfies the strong non-signaling property with respect to any~$C^*$-algebra~$\msB$.
Simply define the positive linear functionals~$\phi_{xa}$ by~$\phi_{xa}(B) := \tr\mleft( \sigma_{xa} B \mright)$ for all~$B\in\msB$.
Thus the algebraic model is at least as general.


We now state the key result of this section.
We will use it as an important component in the proof of our main result in \cref{sec:new-magic}.

\begin{theorem}\label{thm:quantum strong nonsig c*}
For any strongly non-signaling algebraic strategy, the resulting correlation is a (nonlocal) commuting operator correlation, that is, in~$C_{qc}$.
\end{theorem}

In the proof of the theorem we will use the following version of the Radon-Nikodym theorem for $C^*$-algebras, which is well-known to experts in operator algebras.
We refer the reader to \cref{sec:operator algebra prelims} for the concepts used in its statement.

\begin{proposition}[Radon-Nikodym theorem for $C^*$-algebras]\label{prop:Blackadar}
Let~$\phi$ and~$\psi$ be positive linear functionals on a unital $C^*$-algebra~$\msB$ with~$\psi\leq \phi$.
Then there exists a unique operator $T\in \pi_{\phi}(\msB)'\in \mcB(\mcH_\phi)$, with $0\leq T\leq \Id$, such that
\begin{equation*}
    \psi(B)
=   \bra{\nu_\phi} T \pi_{\phi}(B) \ket{\nu_\phi},
\end{equation*}
for all $B\in \msB$, where $(\mcH_\phi,\pi_\phi,\ket{\nu_\phi})$ is any GNS triple associated with~$\phi$.
\end{proposition}

The version stated here is \cite[Prop.~II.6.4.6]{blackadar2006operator} and we refer to this reference for a concise~proof.

\begin{proof}[Proof of \cref{thm:quantum strong nonsig c*}]
Observe that~$\phi_{xa} \leq \phi$ for all~$x\in\mcI_A$ and~$a\in\mcO_A$.
Let $(\mcH_\phi, \pi_\phi, \ket{\nu_\phi})$ be a GNS triple associated with~$\phi$.
Then, by \cref{prop:Blackadar}, for each pair $(x,a)$ there exists an operator~$M_{xa} \in \pi_\phi(\msB)'$ (in the commutant) such that $0 \leq M_{xa} \leq \Id$ and we have, for all $B\in\msB$,
\begin{align}\label{eq:state to GNS by RN}
   \phi_{xa}(B) = \bra{\nu_\phi} M_{xa} \pi_{\phi}(B) \ket{\nu_\phi}.
\end{align}
Because~$B_{yb} \in \msB$ for all~$y,b$, it follows that, for all~$x,y,a,b$,
\begin{equation*}
    p(a,b|x,y)
=   \phi_{xa}(B_{yb})
=   \bra{\nu_\phi} M_{xa} \pi_{\phi}(B_{yb}) \ket{\nu_\phi}
=   \bra{\nu_\phi} M_{xa} N_{yb} \ket{\nu_\phi}
\end{equation*}
where $N_{yb} := \pi_{\phi}(B_{yb}) \in \pi_{\phi}(\msB)$.
Moreover, $[M_{xa},N_{yb}] = 0$ because $M_{xa} \in \pi_\phi(\msB)'$.

To conclude that~$p(a,b|x,y)$ is a commuting operator correlation, it remains to argue that $\{M_{xa}\}_{a\in \mcO_A}$ is a POVMs for each~$x\in \mcI_A$ and~$\{N_{yb}\}_{b\in\mcO_B}$ is a POVM for each~$y\in\mcI_B$.
The latter follows from the fact that each~$\{B_{yb}\}_{b\in\mcO_B}$ is a POVM and~$\pi_{\phi}$ is a $*$-homomorphism.
For the former, it suffices to prove that~$\sum_{a\in\mcO_A} M_{xa} = \Id$ for every~$x\in\mcI_A$ because we already know that the operators~$M_{xa}$ are positive.
To this end, we observe that for any two elements $E,F \in \msB$, it holds that
\begin{align*}
    \bra{\nu_\phi} \pi_\phi(E^*) \sum_{a\in\mcO_A} M_{xa} \pi_\phi(F) \ket{\nu_\phi}
&=  \sum_{a\in\mcO_A} \bra{\nu_\phi} M_{xa} \pi_\phi(E^*F) \ket{\nu_\phi} \\
&=  \sum_{a\in\mcO_A} \phi_{xa}(E^*F) \\
&=  \phi(E^*F)
=   \bra{\nu_\phi} \pi_\phi(E^*) \pi_\phi(F) \ket{\nu_\phi}
\end{align*}
where we used that $M_{xa} \in \pi_\phi(\msB)'$ and that~$\pi_\phi$ is a $*$-homomorphism.
Thus we have
\[\bra{\nu_\phi} \pi_\phi(E^*) (\sum_{a\in\mcO_A} M_{xa} - \Id) \pi_\phi(F) \ket{\nu_\phi} = 0\]
for all $E,F \in \msB$.
Since $\pi_\phi(\msB)\ket{\nu_\phi}$ is dense in $\mcH_\phi$, we deduce that $\sum_{a\in \mcO_A} M_{xa}=\Id$, as desired, concluding the proof.
\end{proof}

Finally, any commuting operator strategy for the nonlocal game (\cref{def:qc_strat}) gives rise to a strongly non-signaling quantum strategy.
Simply use the same Hilbert space~$\mcH$, $\sigma_{xa} = \sqrt{M_{xa}} \ketbra\psi \psi \sqrt{M_{xa}}$, the operators~$B_{yb} = N_{yb}$, and let~$\msB$ denote the~$C^*$-algebra generated by the operators~$\{B_{yb}\}_{b\in\mcO_B,y\in \mcI_B}$.
Then it is easily verified that strong non-signaling holds for~$\msB$, noting that $M_{xa} \in \msB'$ (as we started from a commuting operator strategy) and hence the same is true for its positive square roots~$\sqrt{M_{xa}}$.
Altogether we obtain the following corollary which may be of independent interest.

\begin{corollary}\label{cor:eqcharacterization}
The commuting operator correlation set~$C_{qc}$ is equal to the correlations produced by strongly non-signaling algebraic strategies (\cref{def:nonsig a strat seq}), as well as to the correlations produced by (possibly infinite-dimensional) strongly non-signaling quantum sequential strategies (\cref{def:q str nonsig}).
\end{corollary}

As mentioned in \cref{sub:main_results}, this characterization of the quantum commuting operator correlations was independently established in the context of steering~\cite[Cor.~5.3]{banacki2023steering}.

\section{Upper bound on the quantum value of compiled nonlocal games}\label{sec:new-magic}
In this section we prove that the quantum value of a compiled game never exceeds the commuting operator value of the corresponding nonlocal game (\cref{subsec:main}).
Using the same techniques, we also deduce a self-testing result (\cref{subsec:selftest})

\subsection{Upper bound on the quantum value}\label{subsec:main}
The basic idea is as follows.
In \cref{sec:crypto-compiler}, we showed that any QPT strategy of a nonlocal game satisfies an analogue of the strong non-signaling condition discussed in \cref{sec:alternative-characterization}.
More precisely, \cref{prop:polysecurity} states that~\eqref{eq:poly cond} holds to arbitrary precision when the security parameter tends to infinity, for any fixed polynomial in the Bob POVMs.
We would like to take a limit, but as the Hilbert spaces will depend on the security parameter, we instead work with a single \emph{universal} $C^*$-algebra.
We can then define a sequence of states on this algebra, which captures precisely all information that can be accessed using the Bob POVMs, for every value of the security parameter, and use compactness of the state space of a $C^*$-algebra to define a limit where the strong non-signaling condition holds exactly.
The result then follows from \cref{thm:quantum strong nonsig c*}. As an application of \cref{thm:compiled-game-bound} we state a result concerning commuting operator self-testing for compiled games.

We now describe the required $C^*$-algebra, denoted~$\msAPOVM^{\mcI_B,\mcO_B}$ for finite sets~$\mcI_B$ and~$\mcO_B$, which is often called the \emph{POVM algebra} \cite{paddock2023operator}.
It has elements~$\{ B_{yb} \}_{y\in\mcI_B,b\in\mcO_B}$ which satisfy the relations~$0 \leq B_{yb} \leq \Id$ and $\sum_{b\in \mcO_B} B_{yb}=1$ for each $y\in \mcI_B$.
Importantly, it satisfies the following \emph{universal property}:
for any Hilbert space~$\tilde\mcH$ and any collection of POVMs~$\{\tilde B_{yb}\}$ on~$\tilde\mcH$, there eyists a unique $*$-homomorphism $\theta \colon \msAPOVM^{\mcI_B,\mcO_B} \to \mcB(\tilde\mcH)$ sending $B_{yb} \mapsto \tilde B_{yb}$ for all $y\in\mcI_B$ and~$b\in\mcO_B$.
The POVM $C^*$-algebras are separable as they are finitely generated.

\begin{theorem}\label{thm:compiled-game-bound}
Let $\mcG$ be any two-player nonlocal game and let~$S$ be any QPT strategy for the compiled game~$\mcGcomp$.
Then it holds that
\begin{align*}
    \limsup_{\lambda\to\infty} \omega_\lambda(\mcGcomp,S) \leq \omega_{qc}(\mcG)
\end{align*}
\end{theorem}

As a direct consequence, we obtain the following upper bound on the (maximal) quantum value of any compiled game (\cref{def:q values compiled game}).

\begin{corollary}\label{cor:comp_val_bound}
For any two-player nonlocal game~$\mcG$, we have~$\omega_{q,\max}(\mcGcomp)\nolinebreak\leq\nolinebreak\omega_{qc}(\mcG)$.
\end{corollary}

We now prove the theorem.

\begin{proof}[Proof of \cref{thm:compiled-game-bound}]
Recall from \cref{eq:p lambda qpt prover} that for each value of the security parameter~$\lambda\in\N$ there exists a Hilbert space~$\mcH_\lambda$, positive operators~$\sigma^\lambda_{xa}$ for~$x\in\mcI_A$ and~$a\in\mcO_A$ such that each~$\sigma^\lambda_x := \sum_{a\in\mcO_A} \sigma^\lambda_{xa}$ is a state, and POVMs~$\{B_{yb}^\lambda\}_{b\in\mcO_B}$ for all~$y\in\mcI_B$, such that the correlations take the following form:
\begin{align*}
    p_\lambda(a,b|x,y)
= \tr(\sigma^\lambda_{xa} B_{yb}^\lambda)
\end{align*}
After passing to a subsequence, we may assume that the limit
\begin{align*}
    \lim_{\lambda\to\infty} \omega_\lambda(\mcGcomp,S)
\end{align*}
exists and is equal to the~$\limsup$ of the original sequence.

The theorem follows if we can show that there exists a further subsequence~$\{\lambda_k\}_{k\in\N}$ such that the correlations~$p_{\lambda_k}$ converge to a commuting operator correlation.
To this end, let $\msAPOVM^{\mcI_B,\mcO_B}$ denote the POVM $C^*$-algebra described above, with its generators~$\{ B_{yb} \}$.
By the universal property, there exist $*$-homomorphisms
\[ \vartheta_{\lambda} \colon\msAPOVM^{\mcI_B,\mcO_B}\to \mcB(\mcH_{\lambda}) \]
such that $\vartheta_{\lambda}(B_{yb})=B_{yb}^{\lambda}$ for all~$\lambda,y,b$.
We can use these to define linear functionals
\begin{align}\label{eq:def phi lambda on c*}
    \phi_{xa}^\lambda \colon \msAPOVM^{\mcI_B,\mcO_B} \to \C, \quad \phi_{xa}^\lambda(\cdot) = \tr(\sigma_{xa}^\lambda \vartheta_{\lambda}(\cdot)).
\end{align}
Observe that each~$\phi_{xa}^\lambda$ is a positive linear functional of norm~$\norm{\phi_{xa}^\lambda} = \phi_{xa}^\lambda(\Id) = \tr(\sigma_{xa}^\lambda) \leq 1$.
Thus we can apply the Banach--Alaoglu theorem (\cref{sec:operator algebra prelims}) to deduce that, for each $x\in \mcI_A, a \in \mcO_A$, the sequence $\{\phi_{xa}^\lambda\}_{\lambda\in\N}$ (and any subsequence thereof) has a weak-$*$ convergent subsequence.
By iteratively passing to convergent subsequences (recall that~$\mcI_A$ and~$\mcO_A$ are finite sets), we obtain a strictly increasing subsequence~$\{\lambda_k\}_{k\in\N}$ and positive linear functionals~$\phi_{xa} \colon \msAPOVM^{\mcI_B,\mcO_B} \to \C$ such that
\begin{align}\label{eq:weak star}
    \lim_{k\to\infty} \phi_{xa}^{\lambda_k}(B) = \phi_{xa}(B)
\end{align}
for every~$x\in\mcI_A$, $a\in\mcO_A$, and $B\in\msAPOVM^{\mcI_B,\mcO_B}$.
Let $\phi_x := \sum_{a\in\mcO_A} \phi_{xa}$.
These are states, because~$\sum_{a\in\mcO_A} \phi_{xa}^\lambda(\Id) = \tr(\sigma^\lambda_x) = 1$ and hence also~$\phi_x(\Id) = 1$, by~\eqref{eq:weak star}.
We now show that~$\phi_x = \phi_{x'}$ for all~$x,x'\in\mcI_A$.
To this end, take any fixed polynomial~$P(\{B_{yb}\})$ in the generators~$B_{yb}$ of $\msAPOVM^{\mcI_B,\mcO_B}$.
Using \cref{eq:def phi lambda on c*,eq:weak star},
\begin{align*}
    \phi_{xa}(P(\{B_{yb}\}))
&=   \lim_{k\to\infty} \phi_{xa}^{\lambda_k}\mleft( P(\{B_{yb}\}) \mright) \\
&=   \lim_{k\to\infty} \tr\mleft( \sigma_{xa}^{\lambda_k} \, \vartheta_{\lambda_k}(P(\{B_{yb}\})) \mright) \\
&=   \lim_{k\to\infty} \tr\mleft( \sigma_{xa}^{\lambda_k} \, P(\{B^{\lambda_k}_{yb}\}) \mright)
\end{align*}
and hence
\begin{align*}
    \phi_x(P(\{B_{yb}\})) = \lim_{k\to\infty} \tr\mleft( \sigma_x^{\lambda_k} \, P(\{B^{\lambda_k}_{yb}\}) \mright).
\end{align*}
Now \cref{prop:polysecurity} implies that~$\phi_x(B) = \phi_{x'}(B)$ for all~$x,x'\in\mcI_A$ and any element of the form~$B = P(\{B_{yb}\})$.
Since these elements are dense in $\msAPOVM^{\mcI_B,\mcO_B}$, it follows that~$\phi_x = \phi_{x'}$ for all~$x,x'\in\mcI_A$.
Thus we have proved that the~$C^*$-algebra~$\msAPOVM^{\mcI_B,\mcO_B}$ along with the functionals~$\{\phi_{xa}\}$ and the operators~$\{B_{yb}\}$ constitute a strongly non-signaling algebraic strategy for the sequential game~$\mcGseq$.
Using \cref{thm:quantum strong nonsig c*}, we obtain that
\[ p(a,b|x,y) = \phi_{xa}(B_{yb}) \]
is a commuting operator correlation.
On the other hand, \cref{eq:weak star} implies that
\begin{align*}
    \lim_{k\to\infty} p_{\lambda_k}(a,b|x,y) = p(a,b|x,y)
\end{align*}
for all~$a,b,x,y$.
It follows that
\begin{align*}
  \lim_{\lambda\to\infty} \omega_\lambda(\mcGcomp, S)
= \lim_{k\to\infty} \omega_{\lambda_k}(\mcGcomp, S)
= \lim_{k\to\infty} \omega(\mcG, p_{\lambda_k})
= \omega(\mcG, p)
\leq \omega_{qc}(\mcG),
\end{align*}
and this concludes the proof of the theorem.
\end{proof}

\subsection{Self-testing in compiled nonlocal games}\label{subsec:selftest}
Many applications of nonlocal games take advantage of the concept of self-testing; prominent examples include \cite{reichardt2013classical,coladangelo2024verifier,ji2020quantum}.
The standard definition is as follows: a \emph{self-test} is a quantum correlation such that any quantum strategy realizing the correlation is equivalent to an \emph{ideal strategy} via local dilations (isometries).
In particular, this implies that in any quantum strategy realizing those correlations, the expectation value of any polynomial in Alice's and Bob's POVM operators is uniquely determined.
This point of view suggests a more abstract definition of self-testing, which has recently been described in~\cite{paddock2023operator}.
It has the advantage that it can also be adapted to the commuting operator setting.

To describe their definition, we first observe that any commuting operator strategy (\cref{def:qc_strat}) gives rise to a state~$\Psi$ on the $C^*$-algebra $\msA := \msAPOVM^{\mcI_A,\mcO_A}\ot_{\max}\msAPOVM^{\mcI_B,\mcO_B}$, via the universal property.
Here, $\otimes_{\max}$ denotes the max tensor product of the two POVM algebras, whose generators we denote by~$\{A_{xa}\}$ and~$\{B_{yb}\}$, respectively (see \cref{subsec:main} for their definition).
Concretetly, this state is uniquely defined by the property that, for any noncommutative polynomial~$P$,
\begin{align}\label{eq:Phi def}
    \Psi(P(\{A_{xa}, B_{yb}\})) = \bra\psi P(\{M_{xa}, N_{yb}\}) \ket\psi.
\end{align}
The converse also holds: for any state on~$\msA$ we can obtain a commuting operator strategy, by the GNS construction.
Hence, a correlation~$p(a,b|x,y)$ is a commuting operator correlation if, and only if, there exists a state~$\Psi$ on $\msA$ such that $\Psi(A_{xa}\otimes B_{yb})=p(a,b|x,y)$ for all $a\in \mcO_A, b\in \mcO_B, x\in \mcI_A$, $y\in \mcI_B$.
This characterization is well known~\cite{junge2011connes,fritz2012tsirelson,ozawa2013tsirelson}.
The insight of \cite{paddock2023operator} is that this gives rise to a natural definition of self-testing.
For example, a \emph{commuting operator self-test} is simply a correlation~$p$ such that there is a \emph{unique} such state~$\Psi$~\cite[Def.~7.1]{paddock2023operator}.%
\footnote{Similarly, the finite-dimensional states on~$\msA$ (i.e., those for which the GNS Hilbert space is finite-dimensional) correspond precisely to the quantum correlations, and one recovers the traditional definition of self-testing by demanding that there is a unique \emph{finite-dimensional} state realizing the given correlations.}
Very often, self-tests arise from nonlocal games.
We make the following definition:

\begin{definition}\label{def:com selftest}
A nonlocal game~$\mcG$ is called a \emph{commuting operator self-test} if any commuting operator strategy~$S$ such that~$\omega(\mcG,S)=\omega_{qc}(\mcG)$ determines the \emph{same} state~$\Psi$ on $\msAPOVM^{\mcI_A,\mcO_A}\ot_{\max}\msAPOVM^{\mcI_B,\mcO_B}$ via \cref{eq:Phi def}.
We say that $\mcG$ is a \emph{commuting operator self-test with corresponding state~$\Psi$} to indicate~$\Psi$ explicitly.
\end{definition}

While less well-known than ordinary self-tests, there are many examples of commuting operator self-tests.
For instance, the CHSH game is not just an ordinary self-test but even a commuting operator self-test, meaning that any optimal commuting operator strategy gives rise to the same (finite-dimensional) state~\cite{paddock2023operator,frei2022quantum}.

To derive our self-testing result, it will be useful to translate this definition into the language of sequential games:

\begin{lemma}\label{lem:translate algebra}
Let $\mcG$ be a commuting operator self-test with corresponding state~$\Psi$.
Let~$\tilde S$ be a strongly non-signaling algebraic strategy for~$\mcGseq$, with positive linear functionals $\phi_{xa}$ and POVMs~$\{\tilde B_{yb}\}$.
If~$\omega(\mcGseq,\tilde S) = \omega_{qc}(\mcG)$, then it holds, for all~$x\in\mcI_A$ and~$a\in\mcO_A$ and for every noncommutative polynomial~$P$, that
\[ \phi_{xa}(P(\{\tilde B_{yb}\})) = \Psi(A_{xa} \ot P(\{B_{yb}\})). \]
\end{lemma}
\begin{proof}
Let $\msB$ denote the $C^*$-algebra generated by the POVM elements~$\{\tilde B_{yb}\}$, so that~$\tilde S$ is strongly non-signaling with respect to~$\msB$.
Let $(\mcH_\phi, \pi_\phi, \ket{\nu_\phi})$ be a GNS triple associated with~$\phi = \sum_a \phi_{xa}$ (which is independent of~$x$).
In \cref{eq:state to GNS by RN} in the proof of \cref{thm:quantum strong nonsig c*} we showed that for any strongly non-signaling algebraic strategy there are POVMs~$\{M_{xa}\}$ in $\pi_\phi(\msB)'$ such that, for all~$B\in\msB$, $\phi_{xa}(B) = \bra{\nu_\phi} M_{xa} \pi_{\phi}(B) \ket{\nu_\phi}$.
In particular, setting $N_{yb} := \pi_\phi(B_{yb})$, we have
 we have
\begin{align*}
    \phi_{xa}(P(\{\tilde B_{yb}\}))
= \bra{\nu_\phi} M_{xa} \pi_{\phi}(P(\{B_{yb}\})) \ket{\nu_\phi}
= \bra{\nu_\phi} M_{xa} P(\{N_{yb}\}) \ket{\nu_\phi}
\end{align*}
and $[M_{xa},N_{by}]=0$ for all~$x\in\mcI_A,y\in\mcI_B,a\in\mcO_A,b\in\mcO_B$.
Thus we have constructed a commuting operator~$S$ strategy that in particular produces the same correlations,
$\phi_{xa}(\tilde B_{yb}) = \bra{\nu_\phi} M_{xa} N_{yb} \ket{\nu_\phi}$,
and hence $\omega(\mcG,S) = \omega(\mcGseq,\tilde S) = \omega_{qc}(\mcG)$.
Because~$\mcG$ is a commuting operator self-test, we must have
\begin{align*}
    \bra{\nu_\phi} M_{xa} P(\{N_{yb}\}) \ket{\nu_\phi}
= \Psi(A_{xa} \ot P(\{B_{yb}\})),
\end{align*}
which concludes the proof.
\end{proof}

We now provide the asymptotic self-testing statement.

\begin{theorem}\label{thm:com self test}
Let $\mcG$ be a commuting operator self-test with corresponding state~$\Psi$.
If~$S$ is a QPT strategy for the compiled game such that $\lim_{\lambda\to\infty} \omega_\lambda(\mcGcomp,S) = \omega_{qc}(\mcG)$, then it holds that
\begin{align*}
    \lim_{\lambda\to\infty} \tr\mleft( \sigma_{xa}^{\lambda} \, P(\{B^{\lambda}_{yb}\}) \mright)
= \Psi(A_{xa} \ot P(\{B_{yb}\}))
\end{align*}
for every~$x\in\mcI_A,a\in\mcO_A$ and for every noncommutative polynomial~$P$.
In particular:
\begin{align*}
    \lim_{\lambda\to\infty} \tr\mleft( \sigma_x^{\lambda} \, P(\{B^{\lambda}_{yb}\}) \mright) = \Psi(P(\{B_{yb}\})).
\end{align*}
\end{theorem}
\begin{proof}
To show the desired convergence, it suffices to show that for any fixed~$x$, $a$, $P$ and for any subsequence~$\{\lambda_\ell\}_{\ell\in\N}$ there is a further subsequence~$\{\lambda_{\ell_k}\}_{k\in\N}$ such that
\begin{align}\label{eq:subseq claim}
  \lim_{k\to\infty} \tr\mleft( \sigma_{xa}^{\lambda_{\ell_k}} \, P(\{B^{\lambda_{\ell_k}}_{yb}\}) \mright)
= \Psi(A_{xa} \ot P(\{B_{yb}\}))
\end{align}
We will show that for any subsequence~$\{\lambda_\ell\}_{\ell\in\N}$ there in fact exists a further subsequence~$\{\lambda_{\ell_k}\}_{k\in\N}$ such that the above holds for all~$x,a,P$.
To this end, note that our assumption implies that $\lim_{\ell\to\infty} \omega_{\lambda_\ell}(\mcGcomp,S) = \omega_{qc}(\mcG)$.
Proceeding like in the proof of \cref{thm:compiled-game-bound}, but with the subsequence~$\{\lambda_\ell\}$ in place of~$\{\lambda\}$, we obtain a further subsequence~$\{\lambda_{\ell_k}\}_{k\in\N}$ and a strongly non-signaling algebraic strategy~$\tilde S$ for the sequential game, given by positive functionals~$\phi_{xa}$ and POVM elements~$\{\tilde B_{yb}\}$, such that
\begin{align}\label{eq:subseq got}
    \lim_{k\to\infty} \tr\mleft( \sigma_{xa}^{\lambda_{\ell_k}} \, P(\{B^{\lambda_{\ell_k}}_{yb}\}) \mright)
= \phi_{xa}(P(\{\tilde B_{yb}\}))
\end{align}
and
\begin{align*}
  \omega(\mcGseq, \tilde S)
= \lim_{k\to\infty} \omega_{\lambda_{\ell_k}}(\mcGcomp, S)
= \lim_{\ell\to\infty} \omega_{\lambda_\ell}(\mcGcomp,S)
= \omega_{qc}(\mcG).
\end{align*}
Now we can apply \cref{lem:translate algebra} to deduce \cref{eq:subseq claim} from \cref{eq:subseq got}.
\end{proof}

\section*{Acknowledgements}
We thank William Slofstra for bringing~\cite{NCPV12} to our attention and Victoria Wright and Lewis Wooltorton for drawing our attention to~\cite{wright2023invertible}.
We also thank Victoria Wright for bringing~\cite{chaturvedi2021characterising,PRXQuantum.2.020334} to our attention and for pointing out that our characterization of commuting operator correlations resolves an open problem in~\cite{wright2023invertible}.
Finally, we thank Michal Banacki for bringing his independent work~\cite{banacki2023steering} to our attention after a preprint of our paper had appeared online.
CP thanks the hospitality at RUB for a productive visit.
MW acknowledges the Simons Institute for the Theory of Computing at UC Berkeley for its hospitality and support.
CP acknowledges funding support from the Natural Sciences and Engineering Research Council of Canada (NSERC).
AK, GM, SS and MW acknowledge support by the Deutsche Forschungsgemeinschaft (DFG, German Research Foundation) under Germany's Excellence Strategy - EXC 2092 CASA - 390781972.
GM is also supported by the European Research Council through an ERC Starting Grant (Grant agreement No.~101077455, ObfusQation).
MW is supported by the European Research Council through an ERC Starting Grant (Grant agreement No.~101040907, SYMOPTIC), by the NWO through grant OCENW.KLEIN.267, and by the BMBF through projects Quantum Methods and Benchmarks for Resource Allocation (QuBRA) and Quantum Algorithm Solver Toolkit (QuSol).


\bibliographystyle{plain}
\bibliography{references}


\end{document}